\documentclass[sigconf]{acmart}
\AtBeginDocument{%
  \providecommand\BibTeX{{%
    \normalfont B\kern-0.5em{\scshape i\kern-0.25em b}\kern-0.8em\TeX}}}

\copyrightyear{2021}
\acmYear{2021}
\setcopyright{acmlicensed}
\acmConference[SIGMOD '21] {Proceedings of the 2021 International Conference on Management of Data}{June 18--27, 2021}{Virtual Event, China}
\acmBooktitle{Proceedings of the 2021 International Conference on Management of Data (SIGMOD '21), June 18--27, 2021, Virtual Event, China}
\acmPrice{15.00}
\acmISBN{978-1-4503-8343-1/21/06}
\acmDOI{10.1145/3448016.3452812.}

\usepackage{booktabs}
\usepackage{balance}  % for  \balance command ON LAST PAGE  (only there!)
\usepackage{xcolor}
\usepackage{amsthm, graphicx,mathalfa,epsfig}
\usepackage{graphicx}
\usepackage{subfig}
\usepackage{hyperref}
\usepackage{url}
\usepackage{wrapfig}
\usepackage[linesnumbered,ruled]{algorithm2e}
\usepackage{algorithmicx}
\usepackage[noend]{algpseudocode}
\usepackage{enumitem}
\usepackage{mathtools}
\usepackage{array}
\usepackage[normalem]{ulem}
\newcommand\pout{\bgroup\markoverwith{\textcolor{black}{\rule[0.5ex]{2pt}{0.4pt}}}\ULon}
\theoremstyle{definition}
\newtheorem{definition}{Definition}[section]
\newtheorem{theorem}{Theorem}[section]

\theoremstyle{remark}

\begin{document}
\fancyhead{}

\title{PCOR: Private Contextual Outlier Release\\via Differentially Private Search}

% \affiliation{%

\author{Masoumeh Shafieinejad}
\affiliation{%
  \institution{University of Waterloo}
  \city{Waterloo}
  \country{Canada}}
\email{masoumeh@uwaterloo.ca}

\author{Florian Kerschbaum}
\affiliation{%
  \institution{University of Waterloo}
  \city{Waterloo}
  \country{Canada}}
\email{fkerschb@uwaterloo.ca}

\author{Ihab F. Ilyas}
\affiliation{%
  \institution{University of Waterloo}
  \city{Waterloo}
  \country{Canada}}
\email{ilyas@uwaterloo.ca}

% \affiliation{%

% \affiliation{%

% \affiliation{%




%\onecolumn
%\thispagestyle{empty} 
%\input{Reviews2}

\begin{abstract}
Outlier detection plays a significant role in various real world applications such as intrusion, malfunction, and fraud detection. Traditionally, outlier detection techniques are applied to find outliers in the context of the whole dataset. However, this practice neglects the data points, namely contextual outliers, that are not outliers in the whole dataset but in some specific neighborhoods. Contextual outliers are particularly important in data exploration and targeted anomaly explanation and diagnosis. In these scenarios, the data owner computes the following information: i) The attributes that contribute to the abnormality of an outlier (metric), ii) Contextual description of the outlier's neighborhoods (context), and iii) The utility \textcolor{black}{score of the context, e.g. its strength in showing the outlier's significance, or in relation to a particular explanation for the outlier.} However, revealing the outlier's context leaks information about the other individuals in the population as well, violating their privacy. We address the issue of population privacy violations in this paper. There are two main challenges in defining and applying privacy in contextual outlier release. In this setting, the data owner is required to release a \emph{valid context} for the queried record, i.e.~a context in which the record is an outlier. Hence, the first major challenge is that the privacy technique must preserve the validity of the context for each record. We propose techniques to protect the privacy of individuals through a relaxed notion of differential privacy to satisfy this requirement. The second major challenge is applying the proposed techniques efficiently, as they impose intensive computation to the base algorithm. To overcome this challenge, we propose a graph structure to map the contexts to, and introduce differentially private graph search algorithms as efficient solutions for the computation problem caused by differential privacy techniques. 
\end{abstract}





\maketitle

\section{Introduction} \label{Introduction}
We borrow the contextual outliers example of Han et al.~\citep{context}: ``The temperature today is $28^{\circ} C$. Is it exceptional (i.e., an outlier)? It depends, for example, on the time and location! If it is in winter in Toronto, yes, it is an outlier. If it is a summer day in Toronto, then it is normal. Unlike traditional outlier detection, in this case, whether or not today's temperature value is an outlier depends on the context --- the date, the location, and possibly some other factors". Contextual outlier detection helps with the problem of \textit{hidden} outliers --- the data points that are considered \textit{normal} compared to the whole population, but are outliers when compared to a subset of the dataset. Consider an analyst performing market research who wishes to determine the companies that are unusually profitable. While Company $V$ might have normal reported profit compared to all companies, it may have outlying profit when compared to media companies with less than 2000 employees. In this case, the released context for Company $V$ is [Business = media] $\land{}$ [$|Employees| < 2000$]. 
For the same data point $V$ there might be multiple contexts \textcolor{black}{(explanations)} in which $V$ is an outlier, each of which might have a different \textit{utility} value.  \textcolor{black}{One example for utility is the size of the population that is covered by the context, which indicates the significance of the outlier. Another example is the context's overlap with a desired context, suitable for scenarios where the private explanation is required to be in a close relation with a chosen context.} The data owner releases a context for $V$ that ideally achieves the highest utility. We assume reporting the outlier data point $V$ is allowed when covered by the business purpose of performing the analysis. However, releasing the context --- although useful for the data analyst --- leaks information about other individuals in the dataset. This might not be permitted without the users' consent, according to privacy laws (e.g. GDPR~\citep{GDPR}). 

\subsection{Technical Challenges}
The privacy for statistical queries of various types, such as average, sum, variance, histogram, median, and maximum likelihood estimator, have been investigated \citep{DP2,DP3,DP4} in the literature. However, the studies of privacy in outlier analysis are few, and mostly -- described in more detail in Section 7 -- addressing the privacy of the outlier \citep{Crowdblend,Outlierprivacy,Bohler,Okada}. Consequently, the privacy of individuals in the remaining population has remained neglected, leaving the private contextual outlier reporting (PCOR) design as an open problem.

Differential privacy~\citep{DP1, DP2, DP5} is a popular privacy notion in data analysis~\citep{Apex} that has shown success in providing privacy when applied to products made by organizations such as the US Census Bureau~\citep{Census1, Census2}, Google~\citep{Rappor}, Apple~\citep{Apple} and Uber~\citep{Uber}. Since differential privacy i) guarantees privacy for individual records, ii) ensures privacy even in the presence of side information about the data, and iii) bounds the information leakage to a total privacy budget ($\epsilon$) across multiple data releases; it thoroughly meets the privacy needs in PCOR. However, the challenges of utilizing differential privacy in PCOR are many, including defining a meaningful notion of differential privacy for contextual outliers, finding a proper technique that fits PCOR's requirements, and proving that the privacy budget is reasonably bounded. When queried about a record, PCOR must select a \emph{private valid} context with ideally high utility among all candidates. To achieve this goal, we propose a variation of the Exponential mechanism \footnote{\textcolor{black}{We introduce the Exponential mechanism in Section 
\ref{ExpMech}.}} in PCOR, as this mechanism provides differential privacy when the nature of the candidates is discrete. However, the conventional definition of privacy and neighboring datasets in the Exponential mechanism does not guarantee the validity of the context in the outlier release case. This motivates us to propose a variation of the Exponential mechanism \textcolor{black}{that guarantees the query response is valid} and provides a relaxed notion of differential privacy. Our proposed technique, similar to the conventional Exponential mechanism, in its formulaic application inevitably enumerates all possible contexts to select one. This requirement heavily affects the performance of the scheme, impeding its practicality. To summarize, our challenge is to provide a design for private contextual outlier reporting that 
\begin{enumerate}
    \item Defines a notion of differential privacy for contextual outliers, and provides a technique for achieving it. As a result, in most cases, the output of PCOR is approximately the same ($\leq e^{\epsilon}$), even if any single record in the input database is arbitrarily added or removed;
    \item Reduces the complexity from exponential time in the formulaic (direct) approach to polynomial time. This reduction however, should not affect whether the design supports differential privacy;
    \item Achieves high utility. Supporting privacy and scalability, should not prevent the design from meeting its purpose in providing \textit{useful} data for the analyst, and;
    \item Is compatible with any utility function (e.g.~population size) as well as any outlier detection algorithm, instead of being designed for a restricted number of specific cases. 
\end{enumerate}

\subsection{Our Contributions}\label{Contribution}
We propose PCOR, a framework for privately releasing contextual outliers, that addresses the aforementioned challenges:
\begin{enumerate}
    \item PCOR utilizes a technique for a relaxed notion of differential privacy that prevents an adversary from learning information about the individuals in the context other than the reported outlier. 
    \item We propose deploying a sampling layer -- based on a graph structure -- prior to applying the Exponential mechanism, to reduce the complexity of PCOR to polynomial time. In our efficient sampling, we analyze the Breadth-First and the Depth-First search algorithms. We modify these algorithms to support differential privacy for the outcome. We contribute the first presentation of differentially private graph search algorithms in the literature.
    \item We prove our proposed solution is private and compare our algorithm in privacy and computational complexity to alternative solutions. The performance and utility of our algorithm is also supported by our experimental results.
    \item We show that our framework can accommodate any utility function for the contexts. \textcolor{black}{We provide results for two examples of utility functions that correspond to two separate classes of outlier release. One, evaluates the context based on the size of the population it covers while the other, evaluates the contexts based on its relation (overlap) with a chosen context.} Ultimately, we claim our algorithm works for the generic case of reporting the contextual outliers and is not limited to any particular outlier detection algorithm. Our claims are reinforced by experimental results over three algorithms for different outlier detection categories. 
\end{enumerate}
PCOR solves the challenge of simultaneously providing privacy, utility, and performance in reporting contextual outliers. Applied on a dataset of 50,000 records, PCOR reduces \textcolor{black}{the runtime from three days} in the direct differentially private approach to 37 minutes; while it maintains 90\% of the maximum utility\footnote{\textcolor{black}{These numbers correspond to the experiments where the utility is ``population size". The experiments for the ``overlap" utility run even faster (Section \ref{Overlap}).}} and guarantees relaxed (output constrained, defined in Section \ref{OCDP}) differential privacy with $\epsilon = 0.2$.
\subsection{Paper Organization}
The rest of this paper is organized as follows: In Section 2 we describe the outlier detection algorithms and define differential privacy. In Section 3 we provide the formal model and problem definitions for our private contextual outlier release (PCOR). We discuss the direct approach of embedding a differential privacy mechanism in contextual outlier detection in Section 4 and address the performance challenge. In Section 5, we introduce a graph representation of contexts used in our efficient sampling algorithms. Privacy proofs of the candidate algorithms are also covered in Section 5. Their performance and accuracy are evaluated through the experiments described in Section 6. Our results in Section 6 confirm that our final candidate succeeds in achieving privacy, accuracy and efficiency. In Section 7, we situate our work in the current body of research, \textcolor{black}{and in Section 8 we conclude our work.} 
\section{Preliminaries}
\subsection{Outlier Detection Algorithms} \label{OutlierAlgs} 
Outlier detection algorithms fall into three main categories: i) Statistics-based, ii) Distance-based and iii) Model-based methods \citep{Ihab}. In statistics-based outlier detection techniques \citep{Grubbs,TMtest,Stats1,ESDtest,Fit2,Fit1,Hist,KDE} it is assumed that normal data points appear in high probability regions of a stochastic model, while outliers appear in its low probability regions. There are two different types of statistics-based approaches: a) Hypothesis testing, and b) Distribution fitting approaches. The hypothesis testing \citep{Grubbs,TMtest, ESDtest} approach calculates a test statistic based on the data points, then it determines whether the null hypothesis claiming no outlier exists in the dataset, is valid or not. \textcolor{black}{We evaluate Grubbs' test~\citep{Grubbs} as a hypothesis testing outlier detection.}
% which expects the data to follow an approximate normal distribution. The test statistic is defined as $G = \frac{max_{i=1, \cdots, N}|Y_i-\Bar{Y}|}{s}$, where $\Bar{Y}$ is the sample mean and $s$ is the standard deviation. The null hypothesis is rejected if $G > G_{critical}$.} 
The distribution fitting approach \citep{Hist,Fit1,Fit2,KDE} first fits a probability density function to the data. Next, it labels the points with low probability as outliers. \textcolor{black}{We select the Histogram approach~\cite{Hist} for the fitting distribution category.}
% to benefit from its strength in not relying on the assumption of knowing an underlying distribution. The histogram method first bins the range of values by dividing the range of the random variable into a series of consecutive and non-overlapping intervals. In the second step, it counts how the number of values fall under each bin. We use equi-width histogram, and choose the square root of the total number of samples as the number of bins. The data points that belong to bins of very low frequency are reported as outliers.} Distance-based outlier detection techniques often define a distance between data points, which is used for defining a normal behavior \citep{LOF, Global,Distance1, Distance2}. The outlier in this setting is a point that is distant from the others.  
We also investigate Local Outlier Factor~\cite{LOF} as the distance-based outlier detection method in our experiments. The local outlier factor (LOF) method scores data points based on the density of their neighboring data points and detects the outlier points based on the scores. 
% The $LOF$ takes values approximately equal to 1 for a normal point versus $ \gg 1$ for an outlier~\cite{LOF} . Hence, the factor can be used to detect outliers. 
% Following paragraphs describe how the algorithm scores the data points. Given a positive integer $k$, the $k-distance$ of an object $p$, denoted by $k-distance(p)$, is defined as the $distance(p, o)$ between $p$ and another object $o$, such that i) there exists at least $k$ objects other than $p$ in the dataset whose distance to $p$ is less than or equal to $distance(p, o)$, and ii) there exists at most $k-1$ objects other than $p$ in the dataset whose distance to $p$ is strictly less than $distance(p, o)$. The $k-distance$ neighborhood of $p$, denoted by $N_k (p)$, contains all objects other than $p$ whose distance to $p$ is less than or equal to $k-distance(p)$. The reachability distance of an object $p$ with respect to object $o$ is defined as $reach-dist_k(p, o) = max\{k-distance(o), d(p, o)\}$. 
% \begin{definition} The local reachability density of $p$ is defined as:
% \begin{equation*}
%     lrd_k(p) = 1/(\frac{\sum_{o \in  N_k(p)} reach-dist_k(p,o)}{|N_k(p)|})
% \end{equation*}
% The local outlier factor of $p$ is defined as:
% \begin{equation*}
%     LOF_k(p) = \frac{\sum_{o \in N_k(p)} \frac{lrd_k(o)}{lrd_k(p)}}{|N_k(p)|}
% \end{equation*}
% \end{definition} 
% The local outlier factor of an outlier candidate $p$, $LOF_k(p)$, takes values approximately equal to 1 for a normal point versus $LOF_k(p) \gg 1$ for an outlier~\cite{LOF}. Hence it can be used to detect outliers.

Contrary to the outlier detection algorithms described above, model-based techniques detect the outliers in a supervised approach. However, in this paper we focus on the unsupervised outlier detection algorithms as they do not require access to labeled data. 

\subsection{Differential Privacy}
Differential privacy (DP), was introduced in 2006 by Dwork et al.~\citep{DP_org}. To protect the individual's data, DP mechanisms require that changing one record in the database should not change the output distribution \textcolor{black}{much. The formal} definition is as follows.
\textcolor{black}{
\begin{definition}{(Differential Privacy~\citep{DP_org})}
A randomized function \textcolor{black}{$\mathcal{M}:\mathcal{D} \rightarrow{} \mathcal{O}:$} gives $\epsilon$-differential privacy if for any set $S \subseteq \mathcal{O}$ and any pair of neighboring datasets $D_1$ and $D_2$ from $\mathcal{D}$ denoted by $(D_1,D_2) \in \mathcal{N}$, such that $D_1$ and $D_2$ differ by adding/removing a record:
\begin{equation}\label{eq:DP}
Pr[\mathcal{M}(D_1) \in S] \leq e^{\epsilon} \times Pr[\mathcal{M}(D_2) \in S]
\end{equation}
\end{definition}}
\subsection{The Exponential Mechanism}\label{ExpMech}
Exponential mechanism is a differential privacy mechanism used for non-numeric values. The goal of this mechanism is to randomly map a set of $n$ inputs each from a domain $\mathcal{D}$ to some output in a range $\mathcal{R}$~\citep{ExpMech}. There are no specific assumptions about the nature of $\mathcal{D}$ or $\mathcal{R}$ other than a base measure $\mathcal{\mu}$ on $\mathcal{R}$ exists. We proceed with the common assumption that the base measure $\mathcal{\mu}$ is uniform. The Exponential mechanism is driven by a utility function \textcolor{black}{$u : \mathcal{D} \times \mathcal{R} \rightarrow{} \mathbb{R}$} that assigns a real valued score to any pair $(D, r)$ from $\mathcal{D} \times \mathcal{R}$, higher scores indicate more appealing $r$'s for $D$. Given a $D \in \mathcal{D}$ the goal of the mechanism is to return an $r \in \mathcal{R}$ such that $u(D, r)$ is (approximately) maximized while guaranteeing differential privacy. To provide the definition for the Exponential mechanism, we need to define \textit{sensitivity} first.
\begin{definition}\label{def:Sensitivity} (Sensitivity) 
The sensitivity of a function \textcolor{black}{$u : \mathcal{D} \times \mathcal{R} \rightarrow{} \mathbb{R}$,} is defined as \textcolor{black}{\[\Delta u = max_{(D_1,D_2) \in  \mathcal{N}, r \in \mathcal{R}} = |u(D_1,r)-u(D_2,r)|.\]} The sensitivity of $u$ is the maximum change in the function output resulted from replacing any individual's data with a different one. In other words, sensitivity is the maximum change in $u$, if we switch from any input dataset $D_1$ to any neighbor $D_2$ of $D_1$. 
\end{definition}
\begin{definition}{(The Exponential mechanism)}
~\citep{ExpMech} For any function \textcolor{black}{$u : \mathcal{D} \times \mathcal{R} \rightarrow{} \mathbb{R}$,} the Exponential mechanism is defined as follows:
\begin{multline}
    Exp^\epsilon_u(D, \mathcal{R}) := \text{Choose }r, \\ Pr[r] = \frac{exp(\frac{\epsilon u(D, r)}{2\Delta u})}{\sum_{r' \in \mathcal{R}} exp(\frac{\epsilon u(D, r')}{2\Delta u})}  
\end{multline}
\end{definition}
\textcolor{black}{We include Theorem \ref{thm: ExpPriv} for the privacy of the Exponential mechanism and refer to McSherry et al.'s work~\citep{ExpMech} for the its proof. For PCOR's accuracy, we present empirical evaluations in Section \ref{Experiments}. }
\begin{theorem}
Theorem 1. (Privacy~\citep{ExpMech}). The Exponential mechanism $Exp^\epsilon_u(d)$ gives $(2\epsilon \Delta u)$-differential privacy.
\label{thm: ExpPriv}
\end{theorem}
% \begin{theorem}\label{Thm: ExpAcc}
% (Accuracy~\citep{ExpMech}). For those values of $t$ satisfying  \textcolor{black}{$t\geq \frac{ln(\frac{OPT}{t \mu(S_t)})}{\epsilon}$}, where $OPT$ stands for $max_r u(d, r)$ and $S_t = \{r : u(d, r) > OPT - t\}$, the inequality $E[u(D, Exp^\epsilon_u(D))] \geq OPT - 3t$ holds. 
% \end{theorem}
As shown in Theorem \ref{thm: ExpPriv}, for the Exponential mechanisms to be the most useful, sensitivity ($\Delta u$) must be limited. Commonly the utility function is chosen such that $\Delta u \leq 1$, so that $Exp^\epsilon_u$ ensures $(2\epsilon)$-differential privacy. 
\subsection{Output Constrained Differential Privacy} \label{OCDP}
To provide privacy in our PCOR, we utilize a relaxed notion of differential privacy, namely Output Constrained Differential Privacy~\citep{OCDP}. We postpone the reasoning for this choice to Section \ref{OCDPinPCOR}, after introducing the notations and defining the problem at the beginning of Section \ref{Problem}.
\theoremstyle{definition}
\begin{definition}{($f$-Neighbors)}\label{def:fneighbor}
Given function $f:\mathcal{D} \rightarrow{} O$, for any \textcolor{black}{pair} of datasets $D_1,D_2 \in \mathcal{D}$, \textcolor{black}{the} datasets $D_1$ and $D_2$ are neighbors w.r.t.~$f$, denoted by $\mathcal{N}(f(.))$, if 
\begin{enumerate}
    \item $f(D_1) = f(D_2) \neq \emptyset$
    \item $(D_1, D_2) \in \mathcal{N}$
    %\item there is no database $D_3 \in \mathcal{D}$, where $f(D_1) = f(D_3)$, such that $\Delta(D_1, D_2) \subset \Delta(D_1, D_3)$
\end{enumerate}
\end{definition}
\textcolor{black}{The $f$-Neighbors $D_1$ and $D_2$, not only have to satisfy the general neighboring condition and differ from each other in only one record, but also need to result in the same (non-empty) output result/set when the function $f$ is applied to them. In our outlier release use case, the function $f$ outputs the set of all valid contexts in the dataset for a particular outlier. We elaborate outlier release in Section \ref{Problem} and define $f$ in Definition \ref{def: COE}.}
\begin{definition}{(Output Constrained Differential Privacy \textcolor{black}{ -- OCDP})}\label{def: OCDP}
A randomized mechanism $M$ satisfies $(\epsilon, f)-$OCDP, if for any $(D_1,D_2)$ $\in \mathcal{N}(f(.))$, and every set of outputs $S \subseteq range(M)$, we have: 
\begin{equation}
     Pr[M(D_1) \in S] \leq e^\epsilon Pr[M(D_2) \in S]
\end{equation}
\end{definition}

\section{Problem Definition} \label{Problem}
\textcolor{black}{The outlier release encourages the anomaly or discourages it, due to approval or disapproval incentives respectively. In the approval incentive, the outlier benefits from the announcement as their outlying is desired. However, in the latter, the anomaly is undesired and needs to be proven to penalize the outlier. For approval incentive, imagine a funding agency that awards people with outstanding achievement in underrepresented groups. In this example, the agency is allowed to announce the outlier and the privacy-enhanced context, while any individual in that context is able to deny their participation in the funding competition.
For disapproval incentive, imagine an insurance company increasing the annual payment of a client due to the client's abnormal behavior in a group of clients. While the insurance company is required to provide justification for this increase, they are not allowed to reveal information about any other client. 
\begin{table}
    \centering
    \caption{A sample of income data set $D$}
    \begin{tabular}{ | c | c | c | c | c |}
    \hline
    Record & Job title & City & District & Salary \\ \hline \hline
    1 & Medical Doctor & Montreal & Business & $S_1$ \\ \hline
    2 & Lawyer & Toronto & Business & $S_2$ \\ \hline
    3 & CEO & Ottawa & Diplomatic & $S_3$ \\ \hline
    4 & Lawyer & Toronto & Business & $S_4$ \\ \hline
    5 & Lawyer & Ottawa & Diplomatic & $S_5$ \\ \hline
    6 & Medical Doctor & Toronto & Historic & $S_6$ \\ \hline
    7 & Lawyer & Ottawa & Business & $S_7$ \\ \hline
    8 & Lawyer & Ottawa & Diplomatic & $S_8$ \\ \hline
    9 & CEO & Montreal & Historic & $S_9$ \\ \hline
    10 & Medical Doctor & Toronto & Diplomatic & $S_{10}$ \\
    \hline
  \end{tabular}
  \label{Tab: Example}
\end{table}
To clarify how contextual outlier release violates the privacy of the participants, we provide an example of income analysis with concrete attribute values in Table \ref{Tab: Example}. The goal in this example is to reveal the context that covers the largest population for an outlier. $D$ is a dataset with categorical attributes of $Job title =\{CEO, Medical Doctor, Lawyer\}$, $City$ taking values from $\{Montreal, Ottawa, Toronto\}$, and $District$ with the domain $ \{Business, Historic, Diplomatic\}$, and a numerical attribute Salary. Assume that the data owner reveals the record $V$, e.g. record 8 in Table \ref{Tab: Example}, has its highest anomaly significance in the context $Job title \in \{CEO, Lawyer\}$, in $Diplomatic$ district of $Ottawa$. By having a side information about the individuals in the maximum context, e.g. by knowing there is only one $CEO$ living in $Diplomatic$ district of $Ottawa$, the revealed deterministic statement about the outlier $V$, leaks information about the presence of the $CEO$ in the dataset as well.} 

In our setup, we consider a dataset instance, $D$, of a relational schema, $R$. Attributes of $R$ are presented by the set $attr(R) = \{A_1, \cdots, A_m, M\}$; where $M$ is the \textit{metric} attribute we define the outlier with regard to. We denote by $P_{ij}$, a predicate that is defined over the $j^{th}$ value in the domain of a categorical or numerical attribute $A_i$ derived from $attr(R)$. \textcolor{black}{For instance,  in our running example, $P_{23}$ is the third attribute value in the second attribute in $attr(R) = \{Job title , City, District, Salary\}$, which is \emph{City = Toronto}.} The domain of an attribute $A_i$ includes all possible values for $A_i$, regardless of them being covered in the particular instance $D$; we refer to the domain size of $A_i$ by  $|A_i|$. In reporting private contextual outliers, it is necessary to take all the values in the domain of the attributes into account; the reason for this requirement is discussed in Section 4. A predicate $P_{ij}$ filters the tuples from $D$ to a subset that satisfies the predicate. This filtering can be represented as conjunction of disjunction predicates; $[P_{11} \vee P_{12} \vee \cdots \vee P_{1|A_1|}] \land{} \cdots \land{} [P_{m1} \vee P_{m2} \vee \cdots \vee P_{m|A_m|}]$; this format covers many common types of SQL queries.\\
Now, we can define a  \emph{context}, $C$. A context is represented as a binary vector of the form \textcolor{black}{$\langle c_{11}, \cdots, c_{1|A_1|}, \cdots, c_{m1}, \cdots, c_{m|A_m|} \rangle$} with length $t = \sum^m_{i=1}|A_i|$. The bit $c_{ij}$ in $C$ is set to $1$, if the predicate $P_{ij}$ is covered in the context. $C$ filters the dataset $D$ to the population $D_C$. A tuple $V$ belongs to $D_C$, denoted as $V \in D_C$, if $V$ is selected by a set of predicates $P_{ij} , \forall i \in [1, m]$ in $C$, where $1 \leq j \leq |A_i|$. We denote the total number of tuples in $C$'s population as $|D_C|$. \textcolor{black}{In our example of the dataset with categorical attributes of $attr(R) = \{Job title , City, District, Salary\}$ with the attribute domains of $\{CEO, Medical Doctor, Lawyer\}$ for $Job title$, $\{Montreal, Ottawa, Toronto\}$ for $City$, as well as $District$ attributes $\{Business, Historic, Diplomatic\}$, this combination of predicates: $[P_{11} \vee P_{13}] \land{} [P_{23}] \land{} [P_{32}]$, filters the dataset to the $CEO$s and $Lawyer$s in $Toronto$ who live in the $Historic$ district. The corresponding context to this subset of the dataset is represented by $C= \langle 101001010 \rangle$.} Evidently, any non-empty context should include at least one predicate of each attribute, i.e.~the context vector has a minimum Hamming weight\footnote{The Hamming weight of a binary vector is the number of 1's in the vector.}, $m$. We call context $C'$ connected to $C$, if the Hamming distance of $C'$ and $C$ is 1, i.e.~$C'$ is constructable from $C$ by adding or removing (not both) only one predicate. For our running example of $C= \langle 101001010 \rangle$, being $CEO$s and $Lawyer$s in $Toronto$ who live in the $Historic$ district, the context $C'= \langle 100001010 \rangle$ which is $CEO$s in $Toronto$ who live in the $Historic$ district, is connected to $C$.
Now, we can define the contextual outlier release. Given an outlier detection algorithm that takes as input any population $D_C$ and the metric $M$, and outputs all outlier in $D_C$, we can implement the outlier verification method $f_M(D_C,  V)$ to determine if $V$ is an outlier in $C$ with regard to the metric $M$; formally\\
$f_M(D_C,  V) = $
$\begin{cases}
   true, &\quad\text{V is an outlier in $D_C$ w.r.t. M}\\
   false, &\quad\text{Otherwise}\\
\end{cases}$

We call a context $C$, a \textit{matching context} for an outlier $V$ with regard to metric $M$ if and only if $f_M(D_C,  V) = true$.
\begin{definition}{(Contextual Outlier Enumeration, $COE_M$)} \label{def: COE}
Given a dataset $D$ with schema $R$, a metric attribute M, and an outlier verification function $f_M$ for $M$, $COE_M(D,V)$ produces all matching contexts of $V$; i.e.~all $C$'s from the attributes $attr(R)$ such that $V \in D_C$, and $f_M(D_C,  V) = true$.  
\end{definition}
To the best of our knowledge, the problem $COE_M$ defined above is a hard problem with no existing efficient solution \citep{Enumerate, Enumerate2}. Recall from Section 2.3, that each output of $COE_M$ is assigned a score by a utility function, where higher scores indicate more appealing outputs. \textcolor{black}{We focus on two types of utility, $u_V(D,C)$, of a context $C$ for an outlier $V$: i) a quantitative evaluation of the context independent from other contexts. The population size of $C$ is an very common example for this utility type, since it indicates the outlier's significance. ii) an evaluation of the context that may score the context based on its relation with a chosen context. As an example of this utility type, we investigate the overlap of the candidate context with a fixed starting context -- the results are shown in Section \ref{Overlap}. In either case,} the data owner desires to report $V$ with a context that is both private and achieves high utility. In this work, as implied in the definition, we only consider deterministic outlier detection algorithms embedded in $COE_M$.  
\begin{definition}{(Private Contextual Outlier Reporting)}
\textcolor{black}{Given a dataset $D$ with schema $R$, a metric attribute $M$, a local outlier $V$, and a contextual outlier enumeration $COE_M$ for $M$,} private contextual outlier reporting ($PCOR$) produces a context $C$ for $V$ such that
\begin{enumerate}[label=(\alph*)]
    \item $f_M(D_{C},V) = true$ 
    \item $C$ is produced by a differentially private mechanism
    \item $C$ is expected to maximize utility within $COE_M(D,V)$
    \item $C$ is calculated in $\mathcal{O}(p(n))$
\end{enumerate}
\label{def:PCOR}
\end{definition} 
Note that the context $C$ is the only context that is released to explain the outlying of $V$. \textcolor{black}{All the valid contexts calculated for generating $COE_M$} are solely known by the data owner and are used to compute $C$. \\
In the previous section, we elaborated on the outlier detection algorithms to deploy in $f_M(D_C,  V)$, affecting the results of $COE_M$. We also defined differential privacy's~\citep{DP1} requirement for PCOR and presented the Exponential mechanism for implementing it. In Section \ref{OCDPinPCOR} \textcolor{black}{we explain why} we choose Output Constrained Differential Privacy for PCOR. Subsequently, in Section \ref{Utility}, we discuss PCOR's compatibility with any utility function  \textcolor{black}{through the two examples mentioned above}.

\subsection{Outlier-Preserving Privacy in PCOR} \label{OCDPinPCOR}
As shown in Equation \ref{eq:DP}, a mechanism $\mathcal{M}$ provides differential privacy, if the probability of obtaining any of its possible outcomes $S$ changes with a maximum ratio of $e^\epsilon$, if an input dataset $D_1$ is changed to any neighboring dataset $D_2$. Hence, if there is an $S$ in range of $\mathcal{M}$ that has a zero probability of occurrence for a dataset $D_1$ but a non-zero probability for a neighboring dataset $D_2$, the differential privacy guarantee does not hold. The mechanism must range over all $S \subseteq Range(\mathcal{M})$ with non-zero probability. Since PCOR must output a valid context for $V$ as the final answer, it assigns zero probability to all non-valid contexts. This zero probability is a violation of differential privacy, if
adding or removing a record in a dataset $D_1$ changes a valid context $S$ for $V$ to a non-valid context for $V$ in $D_2$. In PCOR, differential privacy is guaranteed by applying the Exponential mechanism to $COE_M(D ,V)$. $COE_M(D ,V)$, as introduced in Definition \ref{def: COE}, outputs the set of all contexts in which $V$ is an outlier so that $f_M(D_{C},V) = true$ -- property (a) in Definition \ref{def:PCOR} -- is satisfied. Therefore, to guarantee differential privacy as introduced in \textcolor{black}{its original form as represented in} Equation \ref{eq:DP}, for any outlier $V$ and any neighboring datasets $D_1$ and $D_2$, the equality $COE_M(D_1 ,V) = COE_M(D_2 ,V)$ must hold. In other words, adding or removing a record, should not change the set of valid contexts for a particular outlier $V$. \textcolor{black}{This condition is too strict for outlier detection}, and as we show in our experiments in Section \ref{GroupPriv}, there exist several neighboring datasets for various outlier detection algorithms that violate this constraint. \textcolor{black}{Hence, we need some notion of differential privacy with satisfiable requirements for outlier detection algorithms, and we need to show that the shift to this notion, does not sacrifice the privacy guarantee that we are looking for.} Since violating the constraint in the original differential privacy notion is algorithm-dependent, we propose a relaxed notion of (Output Constrained) differential privacy that is algorithm-dependent as well\footnote{Note that the common notion of differential privacy, $(\epsilon, \delta) -DP$, is not algorithm dependent.}. This notion guarantees $\epsilon-$differential privacy when the equality $COE_M(D_1 ,V) = COE_M(D_2 ,V)$ for the outlier detection algorithm holds. We ran several experiments \textcolor{black}{ for the next two observations: i) to what extent the $COE_M(D_1 ,V) = COE_M(D_2 ,V)$ assumption in OCDP matches the results of outlier detection algorithms in practice, ii) what are the effects if this assumption does not hold and whether it results in a privacy sacrifice}. Our results include evaluations of this assumption in group privacy, where the datasets $D_1$ and $D_2$ differ in more than one record. We ran experiments on the outlier detection algorithms introduced earlier in Section \ref{OutlierAlgs}, and provide the results in Section \ref{GroupPriv}.
\subsection{Utility Functions} \label{Utility} 
As described in Section 2.3, the differential privacy level in a mechanism is directly related to the sensitivity of the utility function. We discussed in Section 2.3 that in order to support reasonable levels of privacy, we aim for utility functions with sensitivity close to 1. We introduce two such utility functions here: i) Maximizing the context population size, and ii) Maximizing the overlap with \textcolor{black}{the starting} context.
\subsubsection{Context Population Size}
The outlier verification function, $f$, can return \textit{true} on various contexts for a tuple $V$. Intuitively, a larger context population indicates a higher significance of the outlier.
\begin{definition}{(Maximum Context)}
$C$ is the Maximum context for outlier $V$ with regard to $f_M$ if and only if
\begin{enumerate}[label=(\roman*)]
    \item $f_M(D_C,  V)$ returns true
    \item $\forall C'$ s.t. $f_M(D_{C'},  V)= True$: $|D_C| \geq |D_{C'}|$
\end{enumerate}
\end{definition}
% We require our PCOR to release a context $C'$ for $V$ with greater population, i.e.~closer to the maximum context. However, since the context/output with the highest utility may differ between neighboring database, a differential private mechanism such as PCOR cannot release the output with the highest utility score deterministically. Yet 
By choosing the population size as utility score, the mechanism guarantees the output to be \textit{close enough} to the highest score answer, i.e.~the maximum context. Formally,\\
$u_V(D,C) = $
$\begin{cases}
   -\infty , &\quad f_M(D_C,  V) = false\\
   |D_C|, &\quad\text{Otherwise}\\
\end{cases}$

The utility of $-\infty$ is assigned to non-valid contexts, so that it has zero probability to be picked by the Exponential mechanism. Furthermore, since the population size would differ by at most 1, if replacing the dataset with a neighboring one, the utility function has sensitivity of 1. 
\subsubsection{Overlap with the \textcolor{black}{Starting} Context}
\textcolor{black}{Another subject of interest in outlier analysis is a utility function that scores a context based on its relation to a chosen context. For this category of utility functions, we focus on one that scores a context $C$ based on its population's intersection with the population of \textcolor{black}{a chosen/starting} context for $V$, $C_V$.} Formally,\\
$u_V(D,C) = $
$\begin{cases}
   -\infty, &\quad f_M(D_C,  V) = false\\
   |D_C \cap D_{C_V}| , &\quad\text{Otherwise}\\
\end{cases}$

Note that, changing a record in the dataset $D$ alters the score, i.e.~the intersection size, by at most 1 (sensitivity). As mentioned earlier, PCOR utilizes the Exponential mechanism to provide differential privacy. We start with the direct application of the mechanism and argue why it imposes a substantial computational complexity to PCOR. To resolve the issue, we propose applying a sampling layer prior to deploying the the Exponential mechanism. \textcolor{black}{However,} as we explain in Section \ref{Uniform_Sampling}, sampling in its basic form does not reduce the complexity of the direct approach. To improve the sampling method, we map the contexts and their connections to a graph. Searching over the graph, we observe ``locality" in the matching contexts and use this property to propose efficient sampling methods for our PCOR. A random walk on the graph shows significant superiority over the basic sampling approach. To improve the sampling even further, we use the utility to direct the search and introduce differentially private versions of Depth-First and Breadth-First search algorithms for sampling. This results in a private scalable design for PCOR that achieves high utility. 
\section{Direct Approach for PCOR}\label{Direct}
In the direct application of the Exponential mechanism in OCDP, we apply the mechanism over the outlier enumeration algorithm, as described in Algorithm~\ref{alg:ExpOD}. As we mentioned in Section \ref{Problem}, in order to provide more privacy protection for the individuals it is important that the enumeration algorithm considers all the values in the domain of attributes in $attr(R)$, not just the values covered in the dataset $D$. We show the rationality behind this requirement through the example in  \textcolor{black}{Section \ref{Problem}}. The data owner reports the record $V$ is an outlier in the context $Job title \in \{CEO, Lawyer\}$, in $Diplomatic$ district of $Ottawa$. By enumerating just over the values in the dataset, this report reveals that the individuals in the context are either CEOs or Lawyers. However, by enumerating over all possible values in the domain, \textcolor{black}{the context description will be larger, e.g.  $Job title \in \{CEO, Lawyer, CFO, Diplomat\}$, in $Diplomatic$ district of $Ottawa$. This report leaves it unclear that which of the attribute values CEO, Lawyer, CFO, Diplomat is in the dataset.}
\begin{algorithm}[ht]
    \SetKwInOut{Input}{Input}
    \SetKwInOut{Output}{Output}
    \Input{$D$, $attr(R)$, $V$, $u$, $\epsilon_1$}
    \Output{$C_{p}$}
    $C_M=\emptyset$ \\
    \Comment{Ranging over all possible contexts, to find matches}\\
    \For {$C$ $\subseteq attr(R)$}         
    {\If{$f_M(D_C,  V) = true$}{$C_M \gets C_M \cup C$}}
    \Comment{The Exponential mechanism}\\
    \iffalse
    \For{$c_m \in C_M$}{$Pr(c_m) \gets \frac{e^{\epsilon  u(c_m, D)}}{\sum_{c \in C_M}{e^{{\epsilon_1} u(c, D)}}} $}
    \Comment{Selecting randomly among weighted candidates in $C_M$}\\
    $C_p \sim (C_M, Pr(C_M))$\\ \fi
    return $C_p \gets Exp^{\epsilon_1}_u(D, C_M)$\;
    \caption{PCOR - Direct Approach}
    \label{alg:ExpOD}
    \textbf{Comment:} A context $C$ is a subset of \textcolor{black}{the values of the attributes in $attr(R)$}. All $C$'s included in $COE_M(D,V)$ are candidates for the private output, but with a probability determined by the utility function $u$. The mechanism drawing from the candidates set terminates the algorithm. 
\end{algorithm}
\begin{theorem} Algorithm~\ref{alg:ExpOD} satisfies $(\epsilon=2\epsilon_1, COE_M(. ,V))-OCDP$, according to Definition~\ref{def: OCDP}. 
\label{thm: Pproof1}
\end{theorem}
\begin{proof}
Consider the neighboring pair $(D_1, D_2) \in N(f(.))$, for $COE_M$ being the outlier enumeration used in Algorithm \ref{alg:ExpOD}. According to Definition~\ref{def:fneighbor}, $D_1$ and $D_2$ differ in one record and result in the same range of outputs, $C_M$. To prove the algorithm satisfies the claim in the Theorem~\ref{thm: Pproof1}, we need to show that the probability of algorithm outputting a particular $C^*_p$ for $D_1$ is $2\epsilon_1$-different from $D_2$.
\begin{multline}
    \frac{Pr[Alg1(D_1) = C^*_p]}{Pr[Alg1(D_2) = C^*_p]} = \frac{\frac{e^{{\epsilon_1} u(C^*_p, D_1)}}{\sum_{c \in C_M}{e^{{\epsilon_1} u(c, D_1)}}}}{\frac{e^{{\epsilon_1} u(C^*_p, D_2)}}{\sum_{c \in C_M'}{e^{{\epsilon_1} u(c, D_2)}}}} \\= \frac{e^{{\epsilon_1} u(C^*_p, D_1)}}{e^{{\epsilon_1} u(C^*_p, D_2)}} \times \frac{\sum_{c \in C_M'}{e^{{\epsilon_1}  u(c, D_2)}}}{\sum_{c \in C_M}{e^{{\epsilon_1} u(c, D_1)}}}
\label{eq: basicproof}
\end{multline}
Since $|D_1| - |D_2| = 1$, any context selected for $D_1$ loses in maximum one record when selected for $D_2$. The utility for each context is proportional to its population size, hence \textcolor{black}{its sensitivity is at most 1, i.e. $\Delta u \leq 1$. This means that that for any context in R (including $C_p$) the relation $u(c, D_1) - u(c, D_2) \leq \Delta u \leq 1$ holds. Hence the Equation~\ref{eq: basicproof} simplifies to} 
\begin{equation}
\frac{Pr[Alg1(D_1) = C^*_p]}{Pr[Alg1(D_2) = C^*_p]} \leq e^{\epsilon_1 \Delta u} \times e^{\epsilon_1 \Delta u} = e^{2\epsilon_1 \Delta u} \leq e^{2\epsilon_1}
\label{eq: basicproof2}
\end{equation}
Note that the number of elements in $C_M$ is the same as $C'_M$, due to the neighboring definition for output constrained differential privacy. However, the population size for the contexts in $C_M$ and $C'_M$ might differ in one record. 
\end{proof}  
\begin{theorem} The computation complexity of Algorithm~\ref{alg:ExpOD} is $\mathcal{O}(2^{t+1})$, where $t$ is the total number of attribute values.
\label{thm: Comp1}
\end{theorem}
\begin{proof}
Direct application of the Exponential mechanism requires brute forcing over all possible contexts for $V$ (lines 3-7 in Algorithm~\ref{alg:ExpOD}) and find the matching ones. With a context being a binary vector \textcolor{black}{$\langle c_{11}, \cdots, c_{1|A_1|}, \cdots, c_{m1}, \cdots, c_{m|A_m|} \rangle$} of length $t = \sum^m_{i=1}|A_i|$, this phase is $\mathcal{O}(2^t)$. It also calculates the weight of each candidate (lines 3-7), which is  $\mathcal{O}(2^t)$ as well. Hence, Algorithm~\ref{alg:ExpOD} is $\mathcal{O}(2^{t+1})$.
\end{proof}
\section{Sampling Approach for PCOR}
To resolve the efficiency problem in PCOR, we propose deploying a sampling layer, prior to applying the Exponential mechanism. We start with Uniform Sampling, and analyze its privacy and complexity of the algorithm. 
\subsection{Uniform Sampling}\label{Uniform_Sampling}
The first sampling algorithm we evaluate, picks contexts uniformly from the set of all valid contexts. We apply the Exponential mechanism after sampling. The idea of utilizing Uniform Sampling prior to the Exponential mechanism was investigated by Lantz et al.~\citep{SubExp}. 
\begin{algorithm}
    \SetKwInOut{Input}{Input}
    \SetKwInOut{Output}{Output}
    \Input{$D$ , $attr(R)$, $V$, $u$, $\epsilon_1$, $p$ ($p=\frac{1}{2}$ here)}
    \Output{$C_p$}
    \Comment{Finding $n$ matching contexts}\\
    $C_M=\emptyset$ \\
    \While {$|C_M| \leq n$}         
    {\For{$i \leq t$}{$C[i] =
     \begin{cases}
       \text{1,} &\quad\text{w.p. p}\\
       \text{0,} &\quad\text{w.p. 1-p}\\
     \end{cases}$}
    \If{$f_M(D_C,  V) = true$}{$C_M \gets C_M \cup C$}}
    \Comment{The Exponential mechanism on the samples}\\
    return $C_p \gets Exp^{\epsilon_1}_u(D, C_M)$\;
    \caption{PCOR - Uniform Sampling}
    \label{alg:Unismpl}
    \textbf{Comment:} We form the context vector, by setting one's and zero's randomly. In the general case the probabilities are correspondingly $p$ and $1-p$. We consider $p=\frac{1}{2}$ to achieve Uniform Sampling. We obtain the $C_M$ from the sampling, then apply Exp. mechanism as in Algorithm \ref{alg:ExpOD}.
\end{algorithm}
\begin{theorem}
\textcolor{black}{The} Uniform Sampling described in Algorithm~\ref{alg:Unismpl} satisfies $(\epsilon = 2\epsilon_1, COE_M(. ,V))-OCDP$. 
\label{thm:Unismpl}
\end{theorem}
\begin{proof}
\begin{multline}
    \frac{Pr[Alg2(D_1) = C^*_p]}{Pr[Alg2(D_2)= C^*_p]} \\= \frac{\frac{e^{{\epsilon_1} u(C^*_p, D_1)}}{\sum_{c \in C_M}{e^{{\epsilon_1} u(c, D_1)}}} \times Pr_{D_1}[C^*_p \in C_M]}{\frac{e^{{\epsilon_1} u(C^*_p, D_2)}}{\sum_{c \in C_M'}{e^{{\epsilon_1} u(c, D_2)}}}\times Pr_{D_2}[ C^*_p \in C_M]} \\
\leq e^{{\epsilon_1} \Delta u} \times e^{{\epsilon_1} \Delta u} \times 1 = e^{2{\epsilon_1} \Delta u}
\label{eq: uniproof}
\end{multline}
The first two items in the inequality in Equation~\ref{eq: uniproof} follows the same reasoning as the proof in Theorem~\ref{thm: Pproof1}. The $Pr[C^*_p \in C_M]$ is independent of the database being $D_1$ or $D_2$, as the contexts are chosen based on attribute values not the records, and the attribute values for $D_1$ and $D_2$ are the same. %\footnote{We consider the attribute values in the ``world", not $D_1$ or $D_2$}. 
\end{proof}
\begin{theorem} The computation complexity of Algorithm~\ref{alg:Unismpl} is $\mathcal{O}(2^t)$, where $t$ is the total number of attribute values.
\label{thm: CompUni}
\end{theorem}
\begin{proof}
We calculate the complexity of Algorithm~\ref{alg:Unismpl} in lines 3-10 first. The algorithm keeps sampling among all $2^t$ contexts until it finds $n$ matching contexts for $V$. We want to find out how many contexts the algorithm should sample on average to achieve the goal. Assume the total number of matching contexts for $V$ is $N$; i.e.~the probability of a sample being a matching context is $\frac{N}{2^t}$. Hence, the number of success in sampling follows Binomial distribution $B(x, \frac{N}{2^t})$, with $x$ being the total number of samples~\citep{Probability}. Since the expected value of the number of success in Binomial distribution in this case is $x \times \frac{N}{2^t}$, we need $n \times \frac{2^t}{N}$ samples to obtain $n$ matching contexts on average. Furthermore, Algorithm~\ref{alg:Unismpl} runs the Exponential mechanism on the matching contexts, which adds $O(n)$ complexity. Therefore, Uniform Sampling is $\mathcal{O}(\frac{n \times 2^t}{N} +n)$. Considering $n$ and $N$ being constant values, Uniform Sampling does not effectively change the complexity of the direct approach.  
\end{proof}    
\subsection{Graph-based Sampling} 
We assume the data owner knows a valid \textcolor{black}{starting} context\footnote{The data owner can obtain this context through an initial search.} $C_V$ for a local outlier $V$ and desires to find the private maximum one through PCOR. We map this problem to a search over a graph $G$ \textcolor{black}{initiating form $V$'s starting context} $C_{V}$, aiming to reach the maximum context. We define the context graph with the set of vertices $Vtx = \{C\}$, where $\{C\}$ is all possible contexts defined over the attribute values from $attr(R)$. There is an edge between two contexts $C$ and $C'$ if they are \textit{connected} to each other, i.e.~they are different in the presence of just one attribute value. Every context $C$ can be represented as a binary vector of the form \textcolor{black}{$\langle c_{11}, \cdots, c_{1|A_1|}, \cdots, c_{m1}, \cdots, c_{m|A_m|} \rangle$} with length $t = \sum^m_{i=1}|A_i|$. As a result, flipping a bit in $C$ forms a binary vector for context $C'$ which is different from $C$ in just one (added/removed) attribute value. Hence, there are $t$ number of $C'$s connected to $C$; i.e.~each vertex is of degree $t$. Given that the connected vertices are different in just one attribute value, we hypothesize that if $V$ is an outlier in $C$, then it is more probable to be an outlier in a connected vertex than some randomly chosen vertex among Vtx. In other words, we hypothesize the existence of \textit{locality} in outlier detection algorithms. 
We show in this paper the hypothesis holds for algorithms from any outlier detection category. We start by our first sampling algorithm on the context graph: Random Walk. 
\subsubsection{Random Walk Sampling}
The data owner knows a valid context for a local outlier $V$, namely \textcolor{black}{starting} context $C_V$, and desires to find the private maximum one. In random walk sampling, we initiate from the \textcolor{black}{starting} context and change the attribute values to obtain the next context to continue the chain. The change in the context in each iteration includes adding/removing an attribute value to/from the context.
\begin{algorithm}[ht]
    \SetKwInOut{Input}{Input}
    \SetKwInOut{Output}{Output}
    \Input{$D$ , $attr(R)$, $V$, $C_V$, $u$, $\epsilon_1$}
    \Output{$C_p$}
    $C_M=[C_V]$ \\
    \Comment{$n$ is the total number of samples}\\
    \While {$ |C_M| \leq n$ $\&\&$ $C_{conn} \neq \emptyset$}{    
    $C \gets C_V$\\
    $C_{conn} = $ \{connected vertices to $C$\} \\
    \Comment{$C_i$ is selected randomly from $C_{conn}$}\\
    $C_i \xleftarrow{r} C_{conn}$\\
    \If{$f_M(D_{C_i},  V) = true$}{$C_M \gets C_M \cup C_i$\\ $C_V \gets C_i$}\Else{$C_{conn}.remove(C_i)$}}
    \Comment{The Exponential mechanism on the samples}\\
    \iffalse
    \For{$c_m \in C_M$}{$Pr(c_m) \gets \frac{e^{\epsilon u(c_m, D)}}{\sum_{c \in C_M}{e^{{\epsilon_1} u(c, D)}}} $}
    \Comment{Selecting randomly among weighted candidates in $C_M$}\\
    $C_p \sim (C_M, Pr(C_M))$\\ \fi
    return $C_p \gets Exp^{\epsilon_1}_u(D, C_M)$\;
    \caption{PCOR - Random Walk Sampling}
    \label{alg:Rptsmpl}
    \textbf{Comment:} The random walk adds/removes a random attribute value to/from \textcolor{black}{a starting} context to obtain the next context, the algorithm continues to do so until the next context is a matching one for $V$, this context is the second sample. We continue random walking from the obtained matching context, until we reach the desired number of sampled contexts in the multiset $C_M$. The the Exponential mechanism is applied afterwards as in Alg.~\ref{alg:ExpOD} to the samples to select the final differentially private output. 
\end{algorithm}
\begin{theorem}
The random walk sampling described in Algorithm~\ref{alg:Rptsmpl} satisfies $(\epsilon= 2\epsilon_1, COE_M(. ,V))$--OCDP, where  $\epsilon_1$ is the privacy parameter in the Exponential mechanism. 
\label{thm:RWalk}
\end{theorem}
\begin{proof}
To calculate the privacy provided by an algorithm, we evaluate the probability change of the privacy mechanism outputting a particular value when switching from a database $D_1$ to a neighbor $D_2$. These probabilities are shown by $Pr[Alg3(D) = C^*_p]$ and $Pr[Alg3(D_2) = C^*_p]$ in this case respectively. A closer look to the algorithm, reveals that for $C^*_p$ to be an output for $D_1$, it must be selected by the Exponential mechanism from $C_M$. Hence, 
\begin{multline}
    Pr_{D_1}[Alg3(D_1) = C^*_p]= Pr_{D_1}[C^*_p|C_M] \times Pr_{D_1}[C_M]\\ = Pr_{D_1}[C^*_p|C_M] \times \Pi^{n}_{i=1} Pr_{D_1}[C_i|C_{i-1}] \times  Pr_{D_1}[C_V]
    \label{eq: Chn3}
\end{multline}
Now, we calculate the probability of Algorithm~\ref{alg:Rptsmpl} resulting in the same final answer ($C^*_p$) for the two databases $D_1$, $D_2$. Let $C_M =\{c_1, c_2, \cdots, c_c, \cdots, c_{p}\}$ be the path of contexts sampled by Algorithm~\ref{alg:Rptsmpl} for database $D_1$. To update the $C_M$ from any path of contexts, e.g. $\{c_1, c_2, \cdots, c_i\}$, to another path that is different from the former in only the last context, i.e.~$\{c_1, c_2, \cdots, c_{i+1}\}$, two events need to take place. First, the predecessor context of $c_{i+1}$ be selected from the $i$ elements in the $C_M$, we show this context by $c_j$. Second, among all nodes connected to $c_j$, $c_{i+1}$ be selected. By changing from $D_1$ to $D_2$, the former's probability reduces by $2\epsilon \Delta u$; the latter however maintains the same probability, as the selection among the connected nodes is random. Hence, as we also showed in proving Theorem~\ref{thm:Unismpl}, the probability of selecting the same $C^*_p$ from a set of $n$ samples such as $C_M$, is bounded by $e^{2\epsilon \Delta u}$. Therefore,   
\begin{equation} 
    \frac{Pr[Alg3(D_1) = C^*_p]}{Pr[Alg3(D_2) = C^*_p]} \leq e^{2\epsilon_1 \Delta u}
\label{eq: Smatch}
\end{equation}
Note that multiple paths could lead Algorithm~\ref{alg:Rptsmpl} to achieve the same output for $D_2$ as $D_1$. In our proof however, we consider the worst case where there is only one path, therefore Algorithm~\ref{alg:Rptsmpl} has to follow the same path for $D_2$ that it does for $D_1$. 
\end{proof}
\begin{theorem} The computation complexity of Algorithm~\ref{alg:Rptsmpl} is $\mathcal{O}(t)$.
\label{thm: CompRwalk}
\end{theorem}
\begin{proof}
Random walk starts from \textcolor{black}{a starting context for $V$}, $C_V$ and iteratively changes a random attribute value of $C_V$ until it finds a matching context among the $t$ connected contexts. Afterwards, it forms a path by repeating the same process for the last context on the chain. As searching through connected contexts to a context $C$ is done without replacement, its complexity is $\mathcal{O}(t)$, the process is repeated for a constant number of times $n$ to form $C_M$. Ultimately, the Exponential mechanism is applied to select the final candidate among $n$ (constant number of) samples. Hence, the algorithm's complexity is $\mathcal{O}(nt+n)$. 
\end{proof}  
Exploiting the context graph in the Random Walk algorithm improves the sampling complexity from exponential in Uniform Sampling to linear. Our experimental results in Section \ref{Experiments} also affirms the superiority of the Random Walk in utility and performance over the Uniform Sampling. To improve the utility and performance even more, we alter the walk over the connected contexts from a random method to a one directed by the utility function; in other words, we can improve sampling further by taking the utility already during selection of the next vertex into account. Therefore, we explore the Depth-First and Breadth-First search algorithms for sampling. 
\subsubsection{Depth-First Search Sampling}\label{DPDFS}
We modify the Depth-First search (DFS) algorithm \citep{Algorithms} to provide a differentially private version of the algorithm. We emphasize that differential privacy cannot be supported with a non-modified DFS algorithm, e.g. by perturbing its output. 
As described in Section 2.2, differential privacy requires the algorithm to generate any output with approximately the same ($e^{\epsilon}$ different) probability for neighboring datasets. However, as DFS is a deterministic algorithm, the probability of it generating any output is zero except for one particular value. Hence, applying DFS on a dataset $D_1$ might result in an output $r$ that has probability of zero when DFS is applied on a neighboring dataset $D_2$; this property prevents the original DFS from supporting differential privacy. To overcome the challenge, we modify DFS as presented in the Algorithm~\ref{alg:Depth}. The sampling method initiates a stack with  \textcolor{black}{a starting context for $V$, $C_V$}. Next, it searches over all connected contexts to $C_V$ and selects a matching context by applying the Exponential mechanism to the candidates and pushes the result onto the stack. The last context on the stack is the starting point for the next iteration. The iterations continue until the stack contains the desired number of samples, $n$.
\begin{algorithm}[ht]
    \SetKwInOut{Input}{Input}
    \SetKwInOut{Output}{Output}
    \Input{$D$ , $attr(R)$, $V$, $C_V$, $u$, $\epsilon$}
    \Output{$C_p$}
    $Stack \gets \{C_V\}$, $Samples \gets 0$, $Visited \gets \emptyset $\\
    \Comment{$n$ is the total number of samples}\\
    \While {$|Visited| \leq n$ $\&\&$ $|Stack|>0$}   
    {%\\
    $C \gets Stack.top()$, $Visited \gets Visited \cup C$, $C_{chldn} \gets \emptyset$ \\
    \Comment{Generating all children of C}\\
    \For{$i \leq t$ }{{$C[i] \gets 
    \begin{cases}
       \text{0,} &\quad\text{If $C[i] = 1$}\\
       \text{1,} &\quad\text{If $C[i] = 0$}\\
    \end{cases}$}\\
    \If{$f_M(D_C,  V)= true$ and $C \notin Visited$}{$C_{chldn} \gets C_{chldn} \cup C$}}
    \If{not $C_{chldn}$}{$Stack.pop()$} \Else{
    \Comment{The Exponential mechanism on children  of $C$}\\
    $C \gets Exp^{\epsilon_1}_u(D, C_{chldn})$\\
    \iffalse
    \For{$c_c \in C_{chldn}$}{$Pr(c_c) \gets \frac{e^{\epsilon u(c_c, D)}}{\sum_{c \in C_{chldn}}{e^{\epsilon u(c, D)}}} $}
    $C \sim (C_{chldn}, Pr(C_{chldn}))$\\ \fi
    push $C$ onto $Stack$ \\}}
    \Comment{The Exp. mechanism for the final private answer}\\
    return $C_p \gets Exp^{\epsilon_1}_u(D, Visited)$\;
    \caption{PCOR - Depth-First Search Sampling}
    \label{alg:Depth}
    \textbf{Comment:} The algorithm starts from \textcolor{black}{a starting} context, $C_V$, adds it to the stack. It then calculates all the connected nodes to $C_V$ that are matching contexts for $V$. Next, the algorithm applies the Exponential mechanism to select one of the connected nodes, $c_c$. Afterwards, the $c_c$ is pushed onto the stack, to be the starting point for the next search as well. This procedure continues until $n$ samples are pushed onto the stack. The last sample in the stack is the final privacy preserving answer.
\end{algorithm} 
\begin{theorem}
The differentially private DFS described in Algorithm~\ref{alg:Depth} satisfies $(\epsilon=(2n+2)\epsilon_1$, $COE_M(. ,V))$--OCDP, where $n$ is the total number of samples, $\epsilon_1$ is the privacy parameter in the the Exponential mechanism.
\label{thm:Depth}
\end{theorem}
\begin{proof}
A similar reasoning to the one in Theorem~\ref{thm:RWalk} holds for DFS. The idea is calculating the probability of forming the same stack of nodes for neighboring datasets $D_1$ and $D_2$. However, the probability of obtaining the same stack, i.e. the corresponding probability to $\Pi^{n}_{i=1} Pr_{D_1}[C_i|C_{i-1}] \times  Pr_{D_1}[C_V]$ in Equation~\ref{eq: Chn3}, is not the same value for $D_1$ and $D_2$ anymore. As in differential private DFS, we apply the Exponential mechanism before adding any child to the stack, each conditional probability in Equation~\ref{eq: Chn3} is $\epsilon_1 \Delta u$ different for $D_1$ and $D_2$. Since there are $n$ nodes in the stack, the upper bound is: $\frac{Pr_{D_1}[Stack]}{Pr_{D_2}[Stack]} \leq e^{2n \epsilon_1 \Delta u}$. We also know from Equation~\ref{eq: uniproof}, that the probability of selecting the same $C^*_p$ from a set of $n$ samples ($Stack$), is bounded by $e^{2\epsilon_1 \Delta u}$. This results in
    \begin{equation}
        \frac{Pr[Alg4(D_1) = C^*_p]}{Pr[Alg4(D_2) = C^*_p]} \leq e^{(2n+2) \epsilon_1 \Delta u}
    \label{eq: Depth}
    \end{equation}
\end{proof}
\begin{theorem} The computational complexity of Algorithm \ref{alg:Depth} is $\mathcal{O}(t)$, where $t$ is the total number of attribute values.
\label{thm: CompDFS}
\end{theorem}
\begin{proof}
The differential private DFS algorithm starts a path from \textcolor{black}{a starting context for $V$}, $C_V$, selects the next matching context the $t$ connected contexts by applying the Exponential mechanism to them and pushes the result onto the stack; which forms an $\mathcal{O}(2t)$ step. Then repeats the process for the last context on the stack until the stack size reaches the constant desired value, $n$, times. Hence, the algorithm's complexity is $\mathcal{O}(2nt)$. 
\end{proof} 
\subsubsection{Breadth-First Search Sampling}
We introduce a modified and differentially private version of the  Breadth-First search (BFS) algorithm~\citep{BFS} as our last sampling method. The reason for necessity of modifying the algorithm to support differential privacy is the same as the one provided for Depth-First search algorithm in Section \ref{DPDFS}. Our sampling method starts from \textcolor{black}{a starting context for $V$}, $C_V$. It initiates the set $C_M$ with that context, then searches over all connected contexts to $C_V$ and adds all the matching contexts to $C_M$. To decide the starting point for the next iteration, it applies the Exponential mechanism to the contexts in the $C_M$(treating $C_M$ as a priority queue) and proceeds similarly with the selected context. The process continues until $n$ number of samples are added to $C_M$.
\begin{algorithm}[ht]
    \SetKwInOut{Input}{Input}
    \SetKwInOut{Output}{Output}
    \Input{$D$ , $attr(R)$, $V$, $C_V$, $u$, $\epsilon$}
    \Output{$C_p$}
    $C_M = \{C_V\}$, $Samples \gets 0$, $Visited \gets \emptyset$\\
    \Comment{$n$ is the total number of samples}\\
    \While {$|Visited| \leq n$ $\&\&$ $|C_M|>0$}   
    {%\\
    \Comment{The Exponential mechanism on $C_M$\\
    $C \gets Exp^{\epsilon_1}_u(D, C_M)$\\
    \iffalse
	\For{$c_c \in C_M$}{$Pr(c_c) \gets \frac{e^{\epsilon u(c_c, D)}}{\sum_{c \in C_M}{e^{\epsilon u(c, D)}}} $}
    $C \sim (C_M, Pr(C_M))$\\ \fi
    $Visited \gets Visited \cup C$, $C_M.remove(C)$\\
    \Comment{Adding all children of C to $C_M$}\\
    \For{$i \leq t$ }{{$C[i] \gets
    \begin{cases}
       \text{0,} &\quad\text{If $C[i] = 1$}\\
       \text{1,} &\quad\text{If $C[i] = 0$}\\
    \end{cases}$}\\
    \If{$f_M(D_C,  V) = true$ and $C \notin Visited$}{$C_M \gets C_M \cup C$}}}}
    \Comment{The Exp.  mechanism for the final private answer}\\
    return $C_p \gets Exp^{\epsilon_1}_u(D, Visited)$\;    
    \caption{PCOR - Breadth-First Search Sampling}
    \label{alg:Chnsmpl}
    \textbf{Comment:} The algorithm starts from \textcolor{black}{a starting} context $C_V$, and calculates all the connected contexts to $C_V$, adds the matching ones to the priority queue $C_M$. It continues by applying the the Exponential mechanism to the contexts in $C_M$, and iterates the procedure  until it finds all $n$ samples. 
\end{algorithm} 
\begin{theorem}
The differentially private BFS described in Algorithm~\ref{alg:Chnsmpl} satisfies $((2n+2)\epsilon_1, COE_M(. ,V))$--OCDP, where $n$ is the total number of samples and $\epsilon_1$ is the privacy parameter in the Exponential mechanism. 
\label{thm:Chn}
\end{theorem}
\begin{proof}
Similar to DFS, the privacy proof of BFS follows the proof for Theorem~\ref{thm:RWalk}. Which results in
\begin{equation}
    \frac{Pr[Alg5(D_1) = C^*_p]}{Pr[Alg5(D_2) = C^*_p]} \leq e^{(2n+2) \epsilon_1 \Delta u}
\label{eq: Smatch2}
\end{equation}
\end{proof}
\begin{theorem} The computation complexity of Algorithm~\ref{alg:Chnsmpl} is $\mathcal{O}(t)$, where $t$ is the total number of attribute values.
\label{thm: CompBFS}
\end{theorem}
\begin{proof}
BFS initiates the search from \textcolor{black}{a starting context for $V$}, $C_V$, searches over its connected contexts to find the matching ones to add to $C_M$ which forms an $\mathcal{O}(t)$ step. This step is repeated for a context chosen by the Exponential mechanism from $C_M$ ($\mathcal{O}(nt)$) for the constant desired number of samples, $n$, times. Hence, the algorithm is $\mathcal{O}(n^2t+nt)$.
\end{proof} 
\section{Experiments}\label{Experiments}
We run five groups of experiments to support our choices for the parameters in our PCOR design varying: i) Sampling algorithm, ii) Utility function, iii) Outlier detection algorithms, iv) Privacy parameter $\epsilon$, and v) Group privacy limit in OCDP. 
\subsection{Datasets}
We evaluate PCOR over two datasets. The first one is the Ontario's public sector salary dataset. The province of Ontario (Canada) annually publishes a list of public employees who earn \$100,000 and above~\citep{Ontario}. We filter the dataset to obtain the information of 51000 employees with income higher than 100K in Ontario, the attributes are $attr(R) = $ $\{Job title, Employer, Year, Salary\}$, with domain sizes: $\lvert Job title \rvert= 9$, $\lvert Employer \rvert= 8$, and $\lvert Year \rvert= 8$. The second dataset we use is homicide reports in the United States \citep{Murder}, compiled and made available by the Murder Accountability Project. The data is gathered from multiple agencies and includes the age, sex, ethnicity of victims and perpetrators, their relationship, and the weapon used. We filter the dataset to obtain the information of 110,000 records. The attributes of the data are $attr(R) = $ $\{AgencyType, State, Weapon, VictimAge\}$. The categorical attributes have domains of size 4, 6, and 6 respectively. Our results for both datasets follow the same pattern. Due to space limit, we mainly provide the experiment results on the salary dataset, and refer to our full paper for the detailed experiment results on the homicide dataset.
\subsection{Evaluation Setup} \label{Setup}
We ran our experiments on a machine with 1 TB RAM, 132 cores, Intel(R) Xeon(R) CPU E7-8870 \@ 2.40GHz specifications. We repeated each experiment 200 times to achieve reliable results for PCORs and fair comparison of the algorithms in utility and performance. The performance is measured as the runtime of the algorithm, in two metrics: i) The range of shortest to longest runtime, and ii) The average runtime. The utility is measured as the proportion of the utility of the PCOR's output to the maximum utility. The maximum utility for each data point (outlier), is gained from a reference file. The reference file contains all possible contexts in $attr(R)$ accompanied with their associated utility, and the list of outliers for each context. Generating this file, for our data set of 51,000 records utilizing an optimized program took three days to complete on the machine mentioned earlier. This is also the running time of the direct/formulaic application of Exponential mechanism in OCDP (described in Section \ref{Direct}). From the reference file, we find all matching contexts for outlier $V$ and find the one with maximum utility to divide the utility results of the PCOR by. For the 200 utility results, we provide the mean and the 90\% confidence interval (CI) for each PCOR. We run five groups of experiments to answer the following questions: 1) what is the best (private, accurate, fast) sampling algorithm for PCOR? 2) Is PCOR flexible with embedding other utility functions? 3) Is PCOR flexible with embedding other outlier detection algorithms? \textcolor{black}{We evaluate PCOR over three deterministic outlier detection algorithms, each from a main category described in Section 2.1, namely: Grubbs~\citep{Grubbs}, Local Outlier Factor (LOF) \citep{LOF}, and Histogram~\citep{Hist}.} 4) How does changing each one of privacy, utility and performance parameters affect the other two. 5) How does changing a (group of) record affect the set of possible outcomes, $COE_M$? \\
\subsection{Choosing a Sampling algorithm for PCOR}
We show that using Breadth-First Search algorithm for private sampling results in the highest utility and performance for PCOR. We evaluate four sampling algorithms: i) Uniform Sampling, ii) Random Walk, iii) Breadth-First Search, and iv) Depth-First Search
\footnote{The results are depicted in Figure 1 in the appendix}
. We generate sets of length 50 samples and set $\epsilon = 0.2$ as the total differential privacy budget. This translates to $\epsilon_1 \approx 0.002$ in Depth-First Search and Breadth-First Search as shown in Equations \ref{eq: Depth} and \ref{eq: Smatch2}, and $\epsilon_1 = 0.1$ in Uniform Sampling and Random Walk as shown in Equations \ref{eq: uniproof} and \ref{eq: Smatch} respectively. The run times of the sampling algorithms
\footnote{Figure 1[a-d] in the appendix}
are summarized in Table~\ref{tab:SampComp1}. 
\begin{table}[h!]
  \centering
  \caption{Sampling Methods Comparison - Performance}
  \label{tab:SampComp1}
  \begin{tabular}{cccccc}
    \toprule
    Algorithm & $T_{min}$ & $T_{max}$ & $T_{avg}$ & $\epsilon$ & Outlier\\
    \midrule
    {Uniform}& 7m & 24h & 97m & 0.2 & LOF\\
    {Random Walk} & 15s & 109s & 51s & 0.2 & LOF\\
    {DFS} & 8m & 80m & 40m & 0.2 & LOF\\
    {BFS} & 6m & 61m & 37m & 0.2 & LOF\\
    \bottomrule
  \end{tabular}
\end{table}
Utility is measured as the ratio of the population size of the private answer to that of the maximum context's for an outlier $V$. Utility comparison results are summarized in Table \ref{tab:SampComp2}.
\begin{table}[h!]
  \centering
  \caption{Sampling Methods Comparison - Utility}
  \label{tab:SampComp2}
  \begin{tabular}{ccccc}
    \toprule
    Algorithm & Utility & CI & $\epsilon$ & Outlier\\
    \midrule
    {Uniform}& 0.65 & (0.64, 0.67) & 0.2 & LOF\\
    {Random Walk} & 0.57 & (0.55, 0.60) & 0.2 & LOF\\
    {DFS} & 0.88 & (0.85, 0.90) & 0.2 & LOF\\
    {BFS} & 0.90 & (0.88, 0.93) & 0.2 & LOF\\
    \bottomrule
  \end{tabular}
\end{table}
The Random Walk improves the performance of the uniform sampling significantly, but it sacrifices utility. The BFS and DFS algorithms perform close to each other in running time and utility, they make up for the utility loss in Random Walk by a considerable amount, but require a longer runtime. Their advantage over Uniform sampling in both utility and performance is noticeable. The slight advantage of BFS over DFS in utility and runtime, also holds for PCOR with a different utility function (Section 6.3). Therefore, we select BFS as PCOR's final sampling method.   
\subsection{PCOR and Other Utility Functions}\label{Overlap}
We show the adaptability of PCOR with any utility function by providing the results of embedding another utility function than the population size in Tables \ref{tab:Over1} and \ref{tab:Over2}. 
\begin{table}[h!]
  \centering
  \caption{Intersection Overlap Utility - Performance}
  \label{tab:Over1}
  \begin{tabular}{cccccc}
    \toprule
    Algorithm & $T_{min}$ & $T_{max}$ & $T_{avg}$ & $\epsilon$ & Outlier\\
    \midrule
    %\texttt{Random Walk} & 6s & 56s & 26s & 0.1 & LOF\\
    {DFS} & 3m & 47m & 19m & 0.2 & LOF\\
    {BFS} & 5m & 48m & 20m & 0.2 & LOF\\
    \bottomrule
  \end{tabular}
\end{table}
The utility function in this experiment calculates the population intersection of the private context and the starting context, $C_V$. 
\begin{table}[h!]
  \centering
  \caption{Intersection Overlap Utility - Utility}
  \label{tab:Over2}
  \begin{tabular}{ccccc}
    \toprule
    Algorithm & Utility & CI & $\epsilon$ & Outlier\\
    \midrule
    %\texttt{Random Walk} & 1 & (1, 1) & 0.1 & LOF\\
    {DFS} & 0.88 & (0.86, 0.91) & 0.2 & LOF\\
    {BFS} & 0.97 & (0.95, 0.98) & 0.2 & LOF\\
    \bottomrule
  \end{tabular}
\end{table}
The utility and performance evaluations
\footnote{Figure 2 in the appendix}
confirm that BFS is a superior candidate over the DFS algorithm. 
\subsection{PCOR and Outlier Detection Algorithms}
We show that PCOR can be successfully used with various outlier detection algorithm. We explore two more outlier detection algorithms in our PCOR design: Grubbs, and Histogram  methods. The results are
\footnote{Figure 3 in the appendix}
summarized in Table~\ref{tab:Outlier1} and Table~\ref{tab:Outlier2}. We filter the original dataset to obtain a subset of 11000 records with 14 attributes values in total, and use BFS as the sampling algorithms in the PCOR design. In this set of experiments, the number of samples is $n=50$ and the privacy parameter is $\epsilon=0.2$. 
\begin{table}[h!]
  \centering
  \caption{Outlier Detection Algorithms - Performance}
  \label{tab:Outlier1}
  \begin{tabular}{cccccc}
    \toprule
    Algorithm & $T_{min}$ & $T_{max}$ & $T_{avg}$ & $\epsilon$ & Sampling\\
    \midrule
    {Grubbs} & 0.5m & 1m & 0.8m & 0.2 & BFS\\
    {Histogram} & 2m & 4m & 3.4m & 0.2 & BFS\\
    \bottomrule
  \end{tabular}
\end{table}
For Histogram method, we bin the samples corresponding to a context $C$ to $\sqrt{|D_C|}$ bins, where $|D_C|$ is the size of the population covered by $C$. The bins with less frequency than $2.5 \times 10^{-3} |D_C|$ are labeled as outliers.
\begin{table}[h!]
  \centering
  \caption{Outlier Detection Algorithms - Utility}
  \label{tab:Outlier2}
  \begin{tabular}{ccccc}
    \toprule
    Algorithm & Utility & CI & $\epsilon$ & Sampling\\
    \midrule
    {Grubbs} & 0.86 & (0.84, 0.89) & 0.2 & BFS\\
    {Histogram} & 0.89 & (0.87, 0.91) & 0.2 & BFS\\
    \bottomrule
  \end{tabular}
\end{table}
The results confirm the compatibility of our PCOR design with various outlier detection algorithms. Furthermore, the success of applying BFS on the algorithms implies that the locality (discussed in Section 5.2) exists in all evaluated outlier detection algorithms. \\
% \textbf{Random Outlier Assignment(ROA).} As described in Section 2.1, in addition to our confirming strategy, we develop a negating strategy to examine the hypothesis of locality existence in the outputs of outlier detection algorithms. We deploy a random outlier assignment as introduced in Section 6.1.3, with BFS sampling. We use MD5 hash function that takes a context $C$ and a tuple $V$ as input and announces $V$ an outlier for $C$ if the $MD5(C,V)=0 \mod 128$. The modulus 128 is chosen to guarantee that in this method we have higher expected number of outliers for a context $C$ compared to Grubbs, LOF or Histogram algorithm. Out of 200 test data points, the PCOR with ROA could only form a set of 2 samples in 8 instances and a set of 3 samples in 2 other instances. In the rest of cases, the PCOR with ROA could not add any valid context to the set other than the original context, $C_V$. To compare, all PCORs with other outlier detection algorithms (Grubbs, LOF, Histogram) were able to successfully form sets of 50 samples. Thus, our negating strategy as well reaffirms the premise that there exists a structure in outlier detection algorithm's outputs. 
\subsection{Privacy, Utility, Performance Trade-off}
We show the trade-off between privacy, utility and performance for our final candidate, BFS sampling algorithm, by changing the privacy parameter $\epsilon$ for $n = 50$ number of samples. The results are
\footnote{Figure 4[a-d] in the appendix}
summarized in \textcolor{black}{Table~\ref{tab:Trade1} and Table~\ref{tab:Trade2}.} 
\begin{table}[h!]
  \centering
  \caption{Effect of privacy parameter on performance}
  \label{tab:Trade1}
  \begin{tabular}{cccccc}
    \toprule
    $\epsilon$ & $T_{min}$ & $T_{max}$ & $T_{avg}$ & Sampling & Outlier\\
    \midrule
    {0.05} & 2m & 29m & 15m & BFS & LOF\\
    {0.1} & 2m & 29m & 16m & BFS & LOF\\
    {0.2} & 3m & 30m & 17m & BFS & LOF\\
    {0.4} & 3m & 30m & 17m & BFS & LOF\\
    \bottomrule
  \end{tabular}
\end{table}
\begin{table}[h!]
  \centering
  \caption{Effect of privacy parameter on utility}
  \label{tab:Trade2}
  \begin{tabular}{ccccc}
    \toprule
    $\epsilon$ & Utility & CI & Sampling & Outlier\\
    \midrule
    {0.05} & 0.67 & (0.62, 0.71) & BFS & LOF\\
    {0.1} & 0.82 & (0.78, 0.85) & BFS & LOF\\
    {0.2} & 0.90 & (0.88, 0.93) & BFS & LOF\\
    {0.4} & 0.92 & (0.90, 0.94) & BFS & LOF\\
    \bottomrule
  \end{tabular}
\end{table}
Increasing the $\epsilon$ from 0.05 to 0.4 generally increases the utility, however the parameter $\epsilon=0.2$ is the optimum value and the utility does not increase significantly afterwards. Moreover, changing the privacy parameter does not impose a notable effect on the algorithm runtime. 
\begin{table}[h!]
  \centering
  \caption{Effect of \# of samples on performance}
  \label{tab:Trade11}
  \color{black}\begin{tabular}{cccccc}
    \toprule
    \# Samples & $T_{min}$ & $T_{max}$ & $T_{avg}$ & Sampling & Outlier\\
    \midrule
    25 & 1m & 14m & 7m & BFS & LOF\\
    50 & 3m & 29m & 16m & BFS & LOF\\
    100 & 6m & 61m & 37m & BFS & LOF\\
    200 & 21m & 163m & 99m & BFS & LOF\\
    \bottomrule
  \end{tabular}
\end{table}
\begin{table}[h!]
  \centering
  \caption{Effect of \# of samples on utility}
  \label{tab:Trade22}
  \color{black}\begin{tabular}{ccccc}
    \toprule
    \# Samples & Utility & CI & Sampling & Outlier\\
    \midrule
    {25} & 0.85 & (0.81, 0.88) & BFS & LOF\\
    {50} & 0.88 & (0.85, 0.91) & BFS & LOF\\
    {100} & 0.90 & (0.88, 0.93) & BFS & LOF\\
    {200} & 0.84 & (0.81, 0.87) & BFS & LOF\\
    \bottomrule
  \end{tabular}
\end{table}
We also investigate the effect of changing the number of samples while the privacy parameter is fixed $\epsilon = 0.2$. \textcolor{black}{The results are summarized in Tables \ref{tab:Trade11} and \ref{tab:Trade22}\footnote{Figure 5[a-h] in the appendix}. As we showed in the proof of Theorem \ref{thm: CompBFS}, the only parameters in the computation complexity of BFS are: i) the total number of attribute values $t$, and ii) the number of samples $n$ -- which we referred to as a constant number in the proof. Hence, when $t$ is fixed, $n$ determines the performance of PCOR.} Increasing the number of samples from 25 to 100 decreases the performance, but increases the utility. This increment does not continue for for $n=200$. Recall from Theorem~\ref{thm:Chn}, that keeping a fixed $\epsilon$ in PCOR while increasing $n$, requires using smaller $\epsilon_1$'s in algorithm design, this can cancel out the positive affect of higher $n$'s on utility. \\

\subsection{Context Match and Group Privacy} \label{GroupPriv}
Our reduced salary dataset consists of 11,000 records with 3 attributes that include 14 attribute values in total. The homicide dataset consists of 28,000 records with three attributes that include 12 attribute values in total. We chose small datasets to run several experiments in a reasonable amount of time. We are aware that our results do not benefit from this choice, as changing a single record in a small dataset, more strongly affects the set of outliers than in a large dataset. We repeated each experiment for 100 random outliers and measured the contexts set match of the original dataset and its neighboring datasets. \textcolor{black}{Recall from Section \ref{OCDPinPCOR} that the experiments in this section are to observe: i) to what extent the assumption $COE_M(D_1, V) = COE_M(D_2, V)$ in OCDP matches the results of outlier detection algorithms in practice, ii) what are the effects if this assumption does not hold and whether it results in a privacy sacrifice. For both mentioned datasets and all three outlier detection algorithm, we recorded the $COE_M$ results of the original dataset and its 50 randomly chosen neighboring datasets. The results for the objective i) are summarized in Tables \ref{tab:COEMatchSal} - \ref{tab:COEMatchMur}. These tables} show the results for neighboring datasets that are different from the original dataset ($\Delta D$) in 1, 5, 10, and 25 records, for the salary and the homicide dataset respectively.
\begin{table}[h!]
  \centering
  \caption{\emph{COE} Match - Salary dataset}
  \label{tab:COEMatchSal}
  \begin{tabular}{ccccc}
    \toprule
    Algorithm & $\Delta D = 1$ & $\Delta D = 5$  & $\Delta D = 10$ & $\Delta D = 25$\\
    \midrule
    {Grubbs} & $99.8\%$ & $96.9\%$ & $94.5\%$ & $91.9\%$ \\
    {LOF} & $89\%$ & $87.9\%$ & $86.7\%$ & $85.7\%$ \\
    {Histogram} & $95.5\%$ & $82.1\%$ & $70.8\%$ & $58.8\%$\\
    \bottomrule
  \end{tabular}
\end{table} 
\begin{table}[h!]
  \centering
  \caption{\emph{COE} Match - Homicide dataset}
  \label{tab:COEMatchMur}
  \begin{tabular}{ccccc}
    \toprule
    Algorithm & $\Delta D = 1$ & $\Delta D = 5$ & $\Delta D = 10$ & $\Delta D = 25$  \\
    \midrule
    {Grubbs} & $100 \%$ & $100\%$  & $100\%$ & $97.8\%$\\
    {LOF} & $99.9\%$ & $99.5\%$ & $98.7\%$ & $97.7\%$ \\
    {Histogram} & $98.5\%$ & $85.2\%$ & $69.3\%$ & $53.3\%$ \\
    \bottomrule
  \end{tabular}
\end{table} 
\textcolor{black}{For objective ii) and to evaluate the privacy in the case of non-matching $COE$'s of the neighboring datasets,} we ran another set of experiments. Consider the contexts in the intersection of non-equal sets $COE_M(D_1 ,V)$ and $COE_M(D_2 ,V)$, where $|D_1 - D_2| = 1$. We measured the maximum ratio of the probability of selecting any of these context for $D_1$ to the probability of selecting the same context for $D_2$. This experiment was repeated for 200 outlier samples for three outlier detection algorithm over the Salary dataset. All the experiments confirm this ratio is below $e^{\epsilon}$ (recall the differential privacy requirement in Equation \ref{eq:DP}) for $\epsilon = 0.2$ that is the privacy budget in our experiments that we presented earlier in this section. In other words, for non-matching $COE$'s of neighboring datasets, we found no instance that violates the privacy bound in (unconstrained) $\epsilon-$differential privacy.
\section{Related Work}
The attempts to provide privacy in outlier detection, initiated with using secure computation methods for distributed \citep{Vaidya,Dung} or non-distributed \citep{Alabdul} data. In a different approach, Bhaduri et al. \citep{Nasa} use nonlinear data distortion transformation  and  show how  it  can be useful  for  privacy  preserving  anomaly  detection  from  sensitive datasets. Gehrke et al. \citep{Crowdblend} introduce \textit{crowd blending privacy}, in which for every individual in the database, either the individual blends in to a crowd of $k$ people in the database with respect to the privacy mechanism, or the mechanism ignores the individual's data. Inspired by crowd blending \citep{Crowdblend} and providing different levels of privacy \citep{auction}, Lui and Pass \citep{Outlierprivacy} proposed \textit{tailored differential privacy} for protecting outliers, in which the privacy parameter for an individual is determined by the ``outlierness" of the individual's data in the dataset. B{\"o}hler et al. \citep{Bohler} evaluate the opposite by relaxing the differential privacy definition and granting outliers \textit{less} protection, as they consider the outliers as ``faulty systems or sensors one need to detect". In another attempt in private outlier investigation, Nissim et al. \citep{ClusterPrivacy} use differential to locate a small outlier cluster privately. Okada et al.~in \citep{Okada} break the outlier analysis to two tasks: i) counting outliers in a given subspace and ii) discovering sub-spaces containing many outliers, under the constraints of differential privacy. They show their method achieves better utility compared to the original global sensitivity based methods. Nonetheless, our work is the first to investigate the privacy of the individuals in the outlier's context and not the outlier's privacy.
\section{Conclusion}
The revealed context in contextual outlier release leaks information about the individual records in the dataset. We address this privacy violation and propose techniques for a relaxed notion of differential privacy to provide PCOR and resolve the issue. However, the differential privacy solution in its formulaic application suffers from weak performance, impeding its usage in practice. To achieve efficiency, we propose utilizing a sampling layer in the design. We present differentially private graph search algorithms, first time in the literature, and use them for sampling. We prove PCOR with these sampling methods supports worthwhile levels of differential privacy, while providing the desired utility and performance. We articulate and demonstrate empirically that PCOR is compatible with any utility function in outlier detection. Our results indicate that the relaxation required for privacy in PCOR is a not a strong requirement, and is satisfied in most cases in practice. We also show that PCOR is generic and fits any outlier detection algorithm.  
\section{Acknowledgement}
We gratefully acknowledge the support of NSERC for grants RGPIN-05849, CRDPJ-531191, IRC-537591 and the Royal Bank of Canada for funding this research. This work benefited from the use of the CrySP RIPPLE Facility at the University of Waterloo.

\bibliographystyle{ACM-Reference-Format}
\bibliography{sample-base}

\section{Appendix} 
We present our experiment results in two types of figures: i) histogram of accuracy distribution (Utility), which ranges from 0 to 1, where 1 is the accuracy of the direct approach introduced in Section \ref{Direct} and ii) histogram of runtime distribution (Performance). The range for runtime varies according to the represented experiment.
\begin{figure*}
\subfloat[USample-Utility]{\includegraphics[width=.25\textwidth]{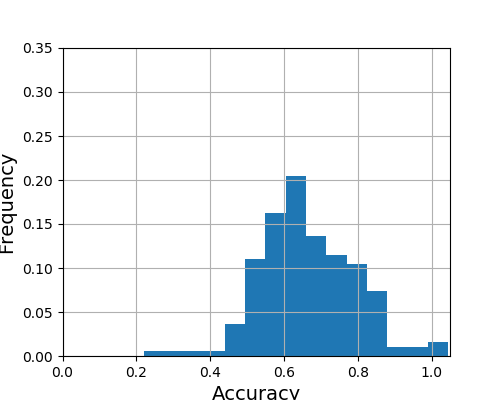}}
\subfloat[Rwalk-Utility]{\includegraphics[width=.25\textwidth]{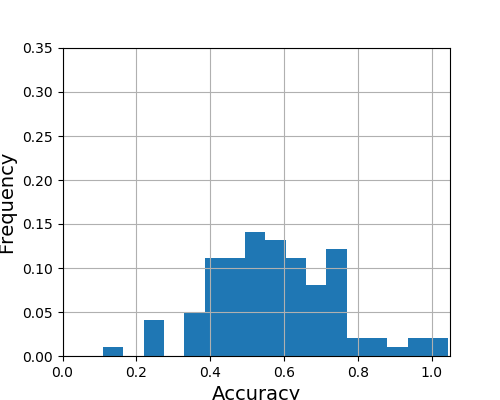}} 
\subfloat[DFS-Utility]{\includegraphics[width=.25\textwidth]{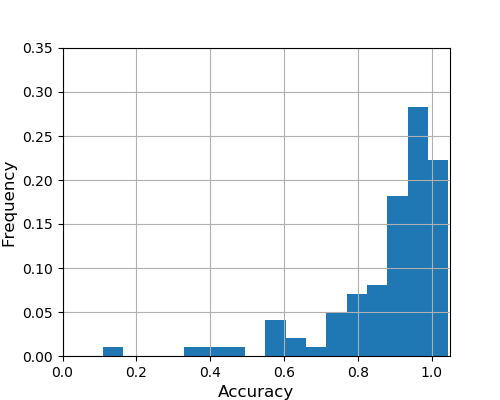}}
\subfloat[BFS-Utility]{\includegraphics[width=.25\textwidth]{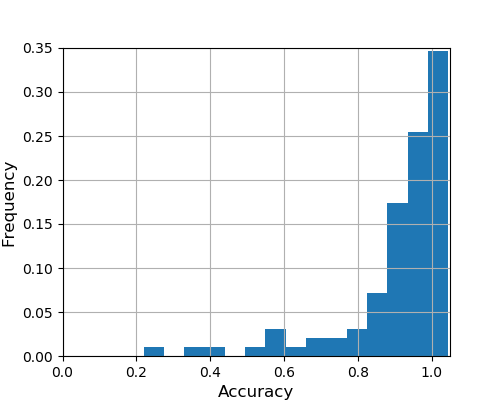}} \hfill
\subfloat[USample-Time]{\includegraphics[width=.25\textwidth]{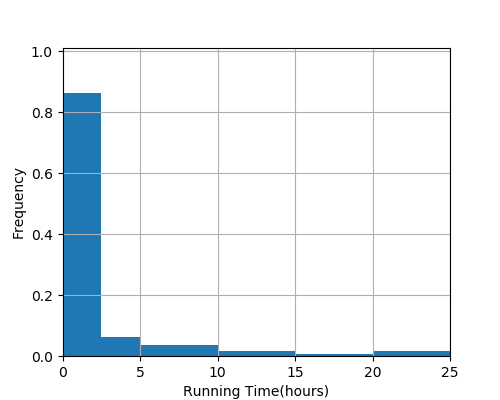}}
\subfloat[Rwalk-Time]{\includegraphics[width=.25\textwidth]{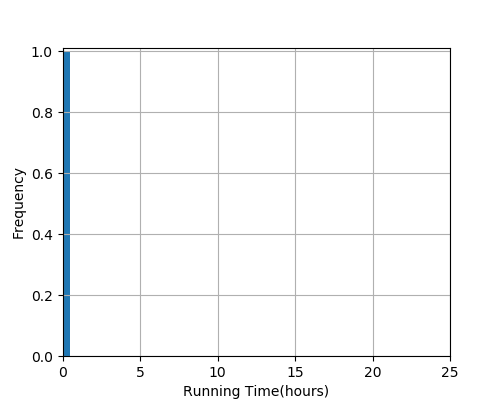}}
\subfloat[DFS-Time]{\includegraphics[width=.25\textwidth]{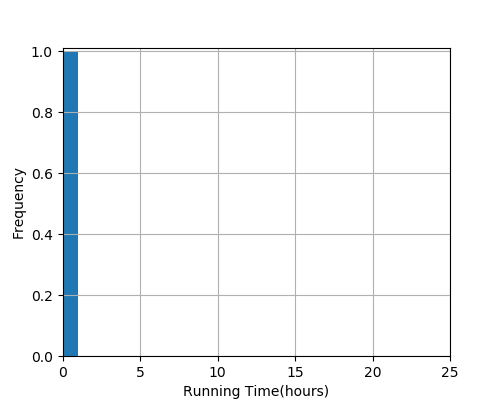}} 
\subfloat[BFS-Time]{\includegraphics[width=.25\textwidth]{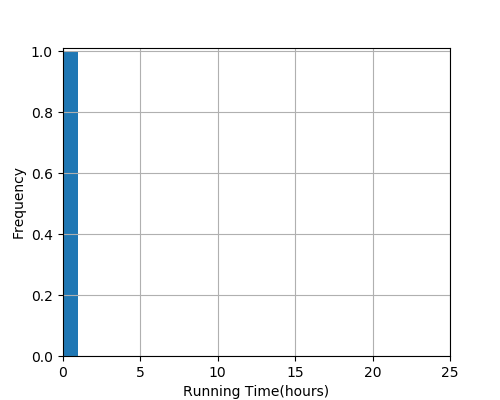}}\hfill
\caption{Utility and Performance of PCORs for different sampling candidates, utility function outputs context population size, outlier detection algorithm is LOF, and $\epsilon = 0.2$ (a,e) Uniform Sampling, (b,f) Random walk, (c,g) DFS, (d,h) BFS; where (a), (b), (c) , (d) represent
utility and (e), (f), (g) ,(h) demonstrate performance in running time}
\label{fig:MaxCtx_ACC_Time} 
\end{figure*}
\begin{figure*}
\centering
%\subfloat[RWalk-Acc]{\includegraphics[width=.25\textwidth]{Figures/URWalk_Acc_50.png}}
\subfloat[DFS-Utility]{\includegraphics[width=.25\textwidth]{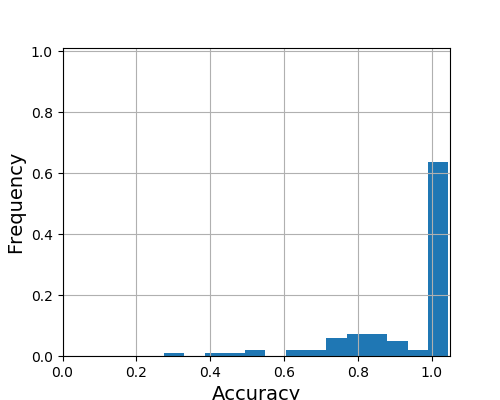}}
\subfloat[BFS-Utility]{\includegraphics[width=.25\textwidth]{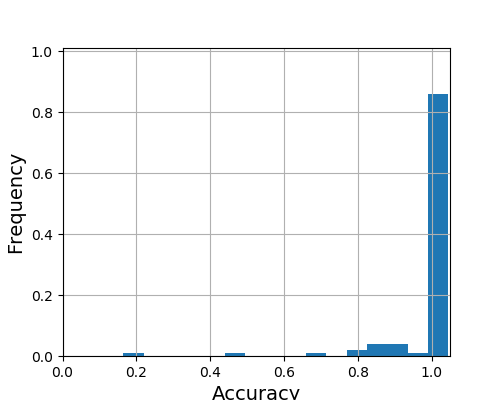}} 
%\subfloat[RWalk-Time]{\includegraphics[width=.25\textwidth]{Figures/URWalk_T_50.png}}
\subfloat[DFS-Time]{\includegraphics[width=.25\textwidth]{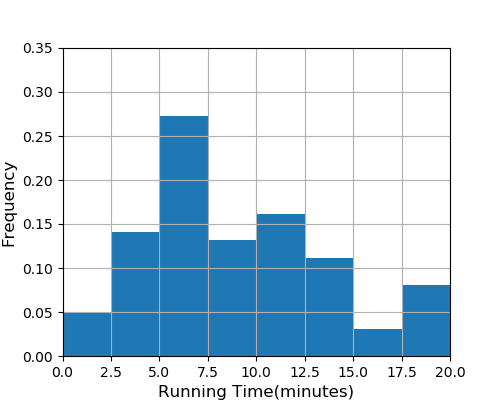}}
\subfloat[BFS-Time]{\includegraphics[width=.25\textwidth]{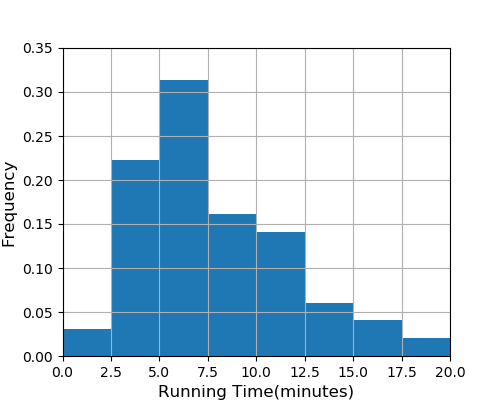}} 
\caption{Utility and Performance PCORS with different sampling candidates, utility is measured by the overlap with $C_V$, LOF is the outlier detection, and $\epsilon =0.1$ (a,c) DFS, (b,d) BFS; where (a), (b) represent utility and (c), (d) demonstrate performance in running time}
\label{fig:Overlap_ACC_Time} 
\end{figure*}

\begin{figure*}
\centering
\subfloat[Grubbs-Utility]{\includegraphics[width=.25\textwidth]{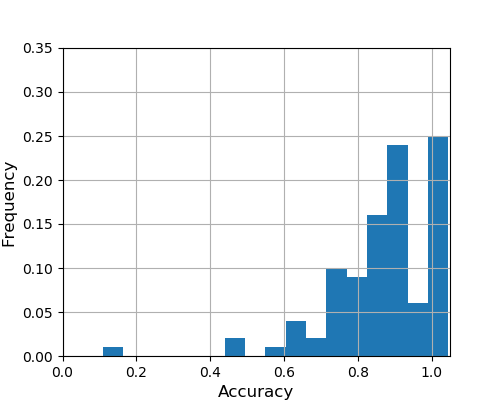}}
%\subfloat[LOF-Acc]{\includegraphics[width=.25\textwidth]{Figures/LOF-Acc.png}}
\subfloat[Histogram-Utility]{\includegraphics[width=.25\textwidth]{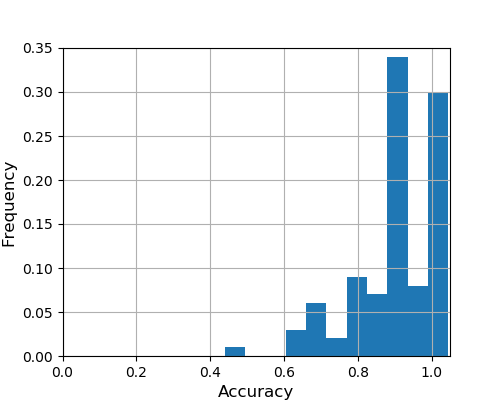}} 
\subfloat[Grubbs-Time]{\includegraphics[width=.25\textwidth]{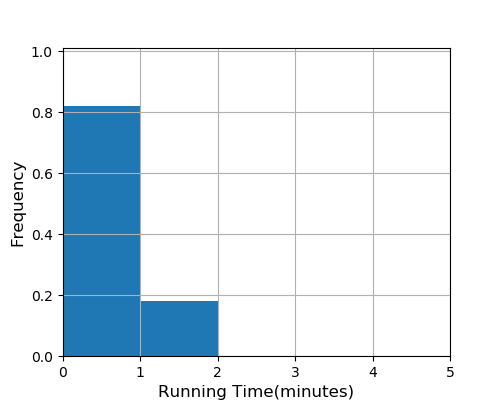}}
%\subfloat[LOF-Time]{\includegraphics[width=.25\textwidth]{Figures/LOF-Time.png}}
\subfloat[Histogram-Time]{\includegraphics[width=.25\textwidth]{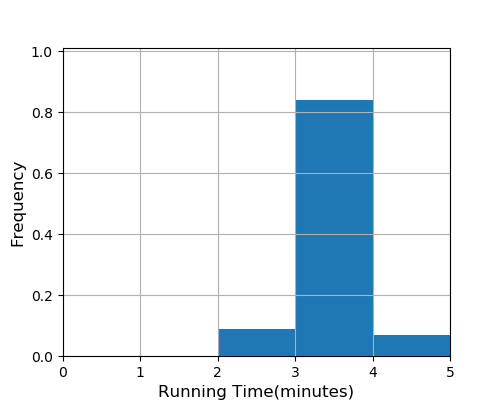}}
\caption{Utility and Performance of PCORs with Grubbs and Histogram outlier detection algorithms using BFS sampling over utility of context population size for $\epsilon = 0.1$ (a,c) Grubbs, (b,d) Histogram; where (a), (b) represent utility and (c), (d) demonstrate performance in running time}
\label{fig:Outlier} 
\end{figure*}

\begin{figure*}
\subfloat[$\epsilon = 0.05$, Utility]{\includegraphics[width=.25\textwidth]{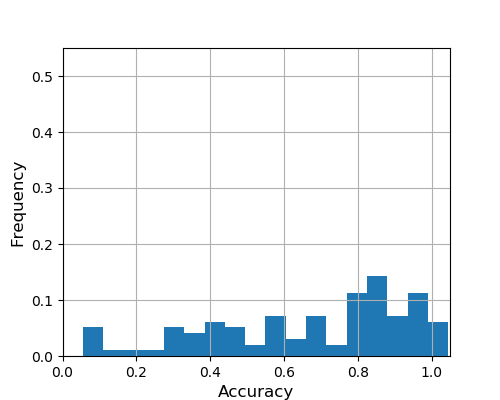}}
\subfloat[$\epsilon = 0.1$, Utility]{\includegraphics[width=.25\textwidth]{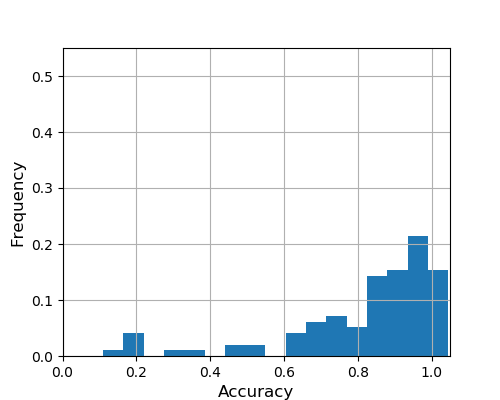}} 
\subfloat[$\epsilon = 0.2$, Time]{\includegraphics[width=.25\textwidth]{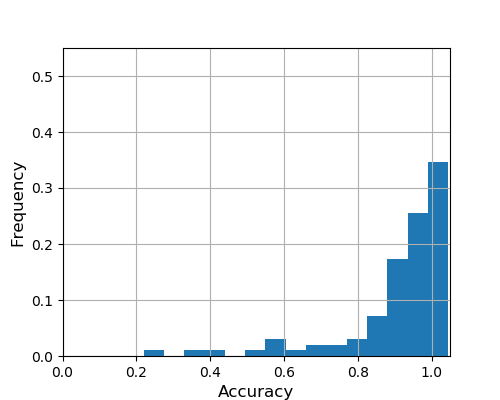}}
\subfloat[$\epsilon = 0.4$, Time]{\includegraphics[width=.25\textwidth]{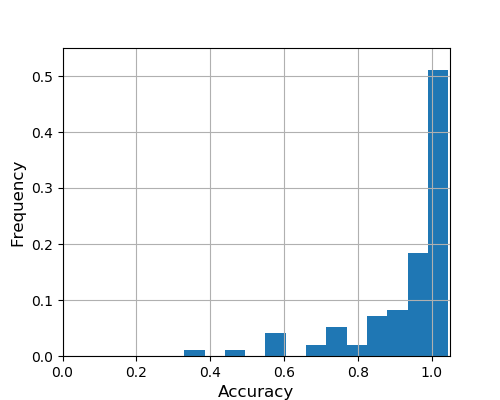}} \hfill
%\subfloat[$\epsilon = 0.05$]{\includegraphics[width=.25\textwidth]{Figures/BFS05_Time.png}}
%\subfloat[$\epsilon = 0.1$]{\includegraphics[width=.25\textwidth]{Figures/BFS1_Time.png}} 
%\subfloat[$\epsilon = 0.2$]{\includegraphics[width=.25\textwidth]{Figures/BFS2_Time.png}}
%\subfloat[$\epsilon = 0.4$]{\includegraphics[width=.25\textwidth]{Figures/BFS4_Time.png}} \hfill
\caption{The effect of privacy parameter over Utility in PCOR with BFS sampling and LOF outlier detection(a) $\epsilon = 0.05$, (b) $\epsilon = 0.1$, (c) $\epsilon = 0.2$, (d) $\epsilon = 0.4$}
\label{fig:Epsilons} 
\end{figure*}

\begin{figure*}
\subfloat[$n =25,$, Utility]{\includegraphics[width=.25\textwidth]{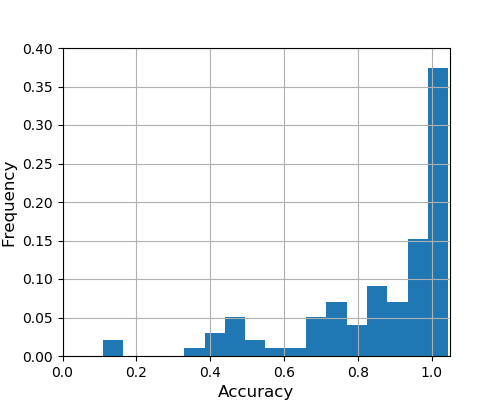}}
\subfloat[$n =50$, Utility]{\includegraphics[width=.25\textwidth]{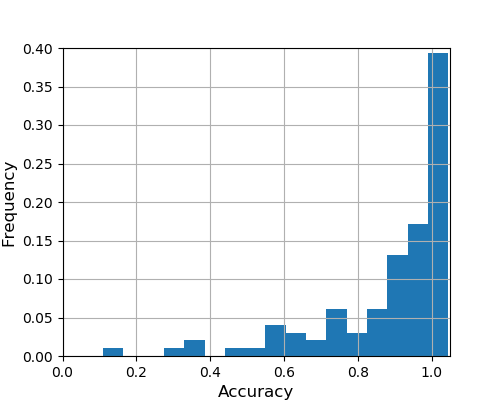}} 
\subfloat[$n =100$, Utility]{\includegraphics[width=.25\textwidth]{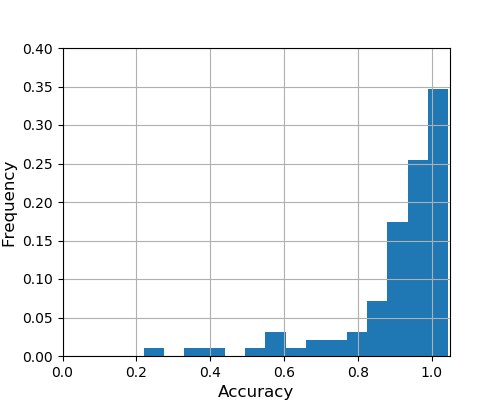}}
\subfloat[$n =200$, Utility]{\includegraphics[width=.25\textwidth]{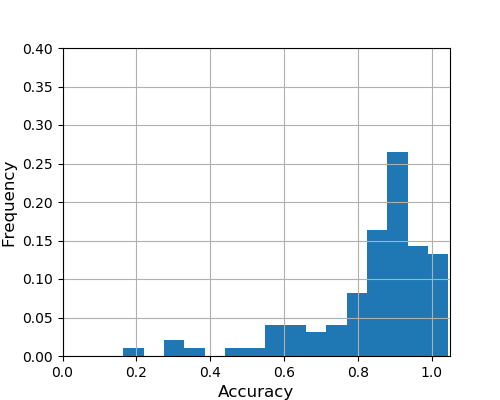}} \hfill
\subfloat[$n =25$, Time]{\includegraphics[width=.25\textwidth]{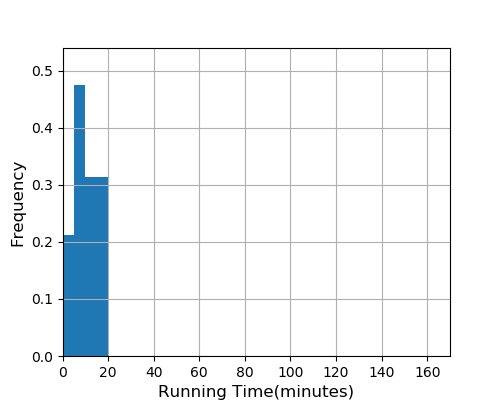}}
\subfloat[$n =50$, Time]{\includegraphics[width=.25\textwidth]{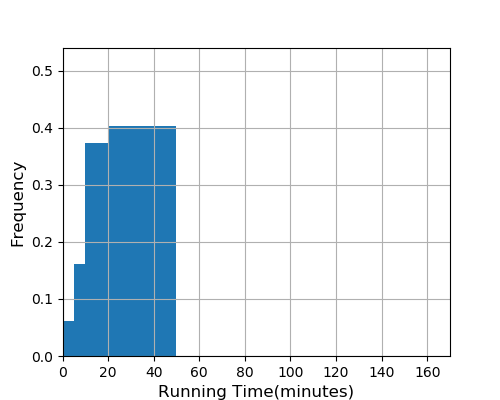}} 
\subfloat[$n =100$, Time]{\includegraphics[width=.25\textwidth]{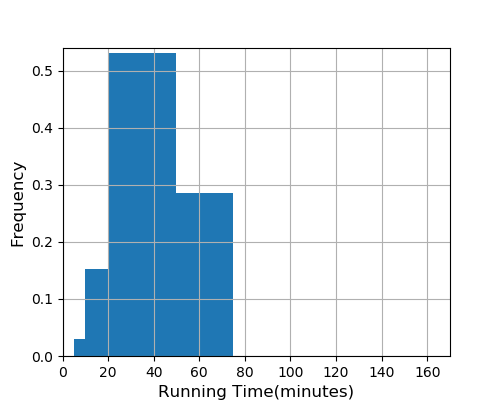}}
\subfloat[$n =200$, Time]{\includegraphics[width=.25\textwidth]{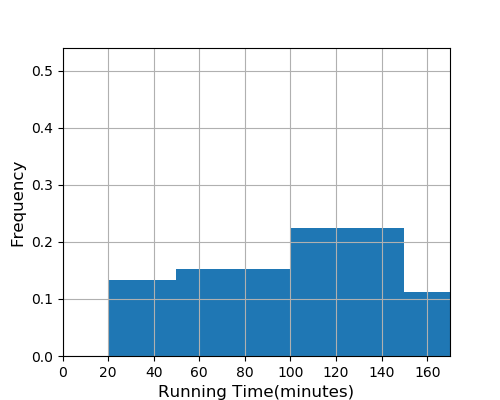}} \hfill
\caption{The effect of number of samples over utility and performance in PCOR with BFS sampling and LOF outlier detection for $\epsilon =0.2$ (a,e) $n = 25$, (b,f) $n = 50$, (c,g) $n = 100$, (d,h) $n=200$}
\label{fig:Epsilons} 
\end{figure*}
\end{document}